%% file: CKM-Arxiv.tex
\documentclass[11pt,letterpaper]{article}
\usepackage{fullpage}
\usepackage{hyperref}

\input{envs}
\hypersetup{colorlinks = false, hidelinks}

\title{Constant Approximation for Capacitated $k$-Median with $(1+\eps)$-Capacity Violation}

\author{G\"{o}kalp Demirci\thanks{Department of Computer Science, University of Chicago, {\tt demirci@cs.uchicago.edu}.} \and
Shi Li\thanks{Department of Computer Science and Engineering, University at Buffalo, {\tt shil@buffalo.edu}. Supported in part by NSF grant CCF-1566356. }
} 

\graphicspath{{./figures/}}

\date{}

\begin{document}
	\maketitle 

	\input{introduction}
	\input{lp}
	\input{setup}
	\input{distribution}

	\input{rounding}

	\bibliographystyle{plain}
	\bibliography{reflist}
	
\end{document}

%% file: envs.tex
\usepackage{amssymb, amsmath}
\usepackage{graphicx}
\usepackage{algorithm}
\usepackage{algorithmic}
\usepackage{xspace}
\usepackage{enumitem}
\usepackage{tabularx}
\usepackage{xcolor}
\usepackage{graphicx} 
\usepackage{tikz}
\usepackage{wrapfig}
\usepackage{standalone}

\newenvironment{properties}[2][1]{
	\begin{enumerate}[label={\normalfont (\ref{#2}\alph*)},labelwidth=*,leftmargin=*,itemsep=0pt,topsep=3pt, start=#1]
	}{
	\end{enumerate}
}

\DeclareMathOperator{\E}{\mathbb E}
\DeclareMathOperator{\union}{\bigcup}

\newcommand{\R}{\mathbb{R}}

\ifdefined\proceeding
\else
\usepackage{amsthm}
\newtheorem{theorem}{Theorem}[section]
\newtheorem{lemma}[theorem]{Lemma}
\newtheorem{corollary}[theorem]{Corollary}
\newtheorem{definition}[theorem]{Definition}

\fi

\newcommand{\set}[1]{\left\{#1\right\}}

\newcommand{\floor}[1]{\left\lfloor#1\right\rfloor}
\newcommand{\ceil}[1]{\left\lceil#1\right\rceil}

\newcommand{\calG}{{\mathcal G}}
\newcommand{\calJ}{{\mathcal J}}

\newcommand{\calP}{{\mathcal P}}
\newcommand{\calS}{{\mathcal S}}

\newcommand{\calV}{{\mathcal V}}

\newcommand{\bbJ}{{\mathbb{J}}}

\newcommand{\eps}{\epsilon}

\newcommand{\tcalS}{{\widetilde \calS}}

\newcommand{\dav}{d_{\mathsf {av}}}
\newcommand{\LP}{\mathsf{LP}}

\newcommand{\diam}{\mathsf{diam}}

\newcommand{\KM}{{\sf KM}\xspace}
\newcommand{\CKM}{{\sf CKM}\xspace}

\newcommand{\MST}{\mathsf{MST}}
\newcommand{\EMD}{\mathrm{EMD}}

\newcommand{\sfN}{{\mathsf{N}}}		
\newcommand{\sfC}{{\mathsf{C}}}

\newcommand{\rma}{{\mathrm{a}}}
\newcommand{\rmb}{{\mathrm{b}}}
	
\newcommand{\remove}{{\mathsf{remove}}}	

\newcommand{\sfd}{{\mathsf{d}}} 

%% file: introduction.tex
\begin{abstract}
	We study the Capacitated $k$-Median problem for which existing constant-factor approximation algorithms are all pseudo-approximations that violate either the capacities or the upper bound $k$ on the number of open facilities. Using the natural LP relaxation for the problem, one can only hope to get the violation factor down to 2. Li [SODA'16] introduced a novel LP to go beyond the limit of 2 and gave a constant-factor approximation algorithm that opens $(1+\epsilon)k$ facilities. 
	
	We use the configuration LP of Li [SODA'16] to give a constant-factor approximation for the Capacitated $k$-Median problem in a seemingly harder configuration: we violate only the capacities by $1+\epsilon$. This result settles the problem as far as pseudo-approximation algorithms are concerned.
\end{abstract}

\section{Introduction}
\label{sec:intro}

In the capacitated $k$-median problem (\CKM), we are given a set $F$ of facilities together with their capacities $u_i \in \mathbb{Z}_{>0}$ for $i\in F$, a set $C$ of clients, a metric $d$ on $F \cup C$, and a number $k$. We are asked to open some of these facilities $F'\subseteq F$ and give an assignment $\sigma : C \rightarrow F'$ connecting each client to one of the open facilities so that the number of open facilities is not bigger than $k$, i.e. $\left| F' \right| \leq k$ (\textit{cardinality constraint}), and each facility $i\in F'$ is connected to at most $u_i$ clients, i.e. $\left| \sigma^{-1}(i)\right| \leq u_i$ (\textit{capacity constraint}). The goal is to minimize the sum of the connection costs, i.e. $\sum_{j\in C}d(\sigma(j), j)$.

Without the capacity constraint, i.e. $u_i = \infty $ for all $i\in F$, this is the famous $k$-median problem (\KM). The first constant-factor approximation algorithm for \KM is given by Charikar et al.\ \cite{CGT99}, guaranteeing a solution within $6\frac{2}{3}$ times the cost of the optimal solution. Then the approximation ratio has been improved by a series of papers \cite{JV01, CG99, AGK01, JMS02, LS13, BPR15}. The current best ratio for \KM is $2.675 +\epsilon$ due to Byrka et al.\ \cite{BPR15}, which was obtained by improving a part of the algorithm given by Li and Svensson \cite{LS13}. 

On the other hand, we don't have a true constant approximation for \CKM. All known constant-factor results are pseudo-approximations which violate either the cardinality or the capacity constraint. Aardal et al.\ \cite{ABG15} gave an algorithm which finds a $(7+\epsilon)$-approximate solution to \CKM by opening at most $2k$ facilities, i.e. violating the cardinality constraint by a factor of 2. Guha \cite{Guha00} gave an algorithm with approximation ratio $16$ for the more relaxed \textit{uniform} \CKM, where all capacities are the same, 
by connecting at most $4u$ clients to each facility, thus violating the capacity constraint by 4. Li \cite{Li14} gave a constant-factor algorithm for uniform \CKM with capacity violation of only $2+\epsilon$ by improving the algorithm in \cite{CGT99}.   For non-uniform capacities, Chuzhoy and Rabani \cite{CR05} gave a 40-approximation for \CKM by violating the capacities by a factor of 50 using a mixture of primal-dual schema and lagrangian relaxations. Their algorithm is for a slightly relaxed version of the problem called \textit{soft} \CKM where one is allowed to open multiple collocated copies of a facility in $F$.  The \CKM definition we gave above is sometimes referred to as \textit{hard} \CKM as opposed to this version. Recently, Byrka et al.\ \cite{BFR15} gave a constant-factor algorithm for hard \CKM by keeping capacity violation factor under $ 3+\epsilon$.

All these algorithms for \CKM use the basic LP relaxation for the problem which is known to have an unbounded integrality gap even when we are allowed to violate either the capacity or the cardinality constraint by $2-\epsilon$. In this sense, results of \cite{ABG15} and \cite{Li14} can be considered as reaching the limits of the basic LP relaxation in terms of restricting the violation factor. In order to go beyond these limits, Li \cite{Li15} introduced a novel LP called the \textit{rectangle} LP and presented a constant-factor approximation algorithm for soft uniform \CKM by opening $(1+\epsilon ) k$ facilities. This was later generalized by the same author to non-uniform \CKM \cite{Li16},  where he introduced an even stronger LP relaxation called the \textit{configuration} LP. Very recently, independently of the work in this paper, Byrka et al.\ \cite{BRU15} used this configuration LP  to give a similar algorithm for uniform \CKM violating the capacities by $1+\epsilon$.

\subsection{Our Result}
In this paper, we use the configuration LP of \cite{Li16} to give an $O(1/\epsilon^5)$-approximation algorithm for non-uniform hard \CKM which respects the cardinality constraint and connects at most \mbox{$(1+\epsilon )u_i$} clients to any open facility $i\in F$.  The running time of our algorithm is $n^{O(1/\epsilon)}$. Thus, with this result, we now have settled the \CKM problem from the view of pseudo-approximation algorithms: either $(1+\epsilon)$-cardinality violation or $(1+\epsilon)$-capacity violation is sufficient for a constant approximation for \CKM.


The known results for the \CKM problem have suggested that designing algorithms with capacity violation (satisfying the cardinality constraint) is harder than designing algorithms with cardinality violation. Note, for example, that the best known cardinality violation factor for non-uniform \CKM among algorithms using only the basic LP relaxation (a factor of $2$ in \cite{ABG15}) matches the smallest possible cardinality violation factor dictated by the gap instance. In contrast, the best capacity-violation factor is $ 3+\epsilon$ due to \cite{BFR15}, but the gap instance for the basic LP  with the largest known gap eliminates only the algorithms with capacity violation smaller than $2$.  
\ifdefined\proceeding 
\else

\fi
Furthermore, we can argue that, for algorithms based on the basic LP and the configuration LP, a $\beta$-capacity violation can be converted to a $\beta$-cardinality violation, suggesting that allowing capacity violation is more restrictive than allowing cardinality violation. 
\ifdefined\proceeding 
We leave the detail to the full version of the paper.
\else
Suppose we have an $\alpha$-approximation algorithm for \CKM that violates the capacity constraint by a factor of $\beta$, based on the basic LP relaxation. Given a solution $(x, y)$ to the basic LP for a given \CKM instance $I$, we construct a new instance $I'$ by scaling $k$ by a factor of $\beta$ and scaling all capacities by a factor of $1/\beta$ (in a valid solution, we allow connections to be fractional, thus fractional capacities do not cause issues).  Then it is easy to see that $(x, \beta y)$ is a valid LP solution to $I'$ (with soft capacities). A solution to $I'$ that only violates the capacity constraint by a factor of $\beta$ is a solution to $I$ that only violates the cardinality constraint by a factor of $\beta$.  Thus, by considering the new instance, we conclude that for algorithms based on the basic LP relaxation, violating the cardinality constraint gives more power. The same argument can be made for algorithms based on the configuration LP: one can show that a valid solution to the configuration LP for $I$ yields a valid solution to the configuration LP for $I'$.  However, this reduction in the other direction does not work: due to constraint~\eqref{LPC:connect-to-open}, scaling $y$ variables by a factor of $1/\beta$ does not yield a valid LP solution.  
\fi
\smallskip

\noindent {\bf Our Techniques.}\ 
Our algorithm uses the configuration LP introduced in \cite{Li16} and the framework of \cite{Li16} that creates a two-level clustering of facilities. \cite{Li16} considered the $(1+\epsilon)$-cardinality violation setting,  which is more flexible in the sense that one has the much freedom to distribute the $\epsilon k$ extra facilities. In our $(1+\epsilon)$-capacity violation setting, each facility $i$ can provide an extra $\epsilon u_i$ capacity; however, these extra capacities are restricted by the locations of the facilities. 
In particular, we need one more level of clustering to form so-called ``groups'' so that each group contains $\Omega(1/\epsilon)$ fractional open facility. Only with groups of $\Omega(1/\epsilon)$ facilities, we can benefit from the extra capacities given by the $(1+\epsilon)$-capacity scaling. Our algorithm then constructs distributions of local solutions. Using a dependent rounding procedure we can select a local solution from each distribution such that the solution formed by the concatenation of local solutions has a small cost.   This initial solution may contain more than $k$ facilities. We then remove some already-open facilities, and bound the cost incurred due to the removal of open facilities.  When we remove a facility, we are guaranteed that there is a close group containing $\Omega(1/\epsilon)$ open facilities and the extra capacities provided by these facilities can compensate for the capacity of the removed facility. 

\smallskip

\noindent {\bf Organization.}\  The remaining part of the paper is organized as follows. In Sections~\ref{sec:config-LP} and \ref{section:setup}, we describe the configuration LP introduced in \cite{Li16} and our three-level clustering procedure respectively. 
 In Section~\ref{sec:distribution}, we show how to construct the distributions of local solutions.  Then finally in Section~\ref{sec:rounding}, we show how to obtain our final solution by combining the distributions we constructed.
 \ifdefined\proceeding
	Due to the page limit, some proofs are omitted and they can be found in the full version of the paper. 
 \fi

%% file: lp.tex
\section{The Basic LP and the Configuration LP}
\label{sec:config-LP}

In this section, we give the configuration LP of \cite{Li16} for \CKM. We start with the following basic LP relaxation:
\bgroup 	\setlength{\abovedisplayskip}{5pt}\setlength{\belowdisplayskip}{0pt}
	\begin{equation}
	\textstyle \min \qquad \sum_{i \in F, j\in C}d(i,j)x_{i,j} \qquad \text{s.t.} \tag{Basic LP} \label{LP:basic}
	\end{equation}
	\setlength{\abovedisplayskip}{0pt}		
	\noindent\begin{minipage}{0.46\textwidth}
		\begin{alignat}{2}
		\textstyle \sum_{i \in F} y_i &\leq  k;  \label{LPC:k-facilities} \\
		\textstyle  \sum_{i \in F}x_{i,j} &=1, &\qquad &\forall j \in C; \label{LPC:client-must-connect} \\
		\textstyle  x_{i,j} &\leq y_i, &\qquad &\forall i \in F, j \in C; \label{LPC:connect-to-open}
		\end{alignat}
	\end{minipage}
	\begin{minipage}{0.49\textwidth}
		\begin{alignat}{2}
		\textstyle  \sum_{j \in C}x_{i,j} &\leq u_iy_i, &\qquad &\forall i \in F; \label{LPC:capacity} \\
		\textstyle  0 \leq x_{i,j}, y_i &\leq 1, &\qquad &\forall i \in F, j \in C. \label{LPC:xy-non-neg} \\
		 \nonumber
		\end{alignat}
	\end{minipage}
\egroup
\vspace*{5pt}

In the LP, $y_i$ indicates whether a facility $i\in F$ is open, and $x_{i,j}$ indicates whether client $j\in C $ is connected to facility $i\in F$. Constraint (\ref{LPC:k-facilities}) is the cardinality constraint assuring that the number of open facilities is no more than $k$. Constraint (\ref{LPC:client-must-connect}) says that every client must be fully connected to facilities. Constraint (\ref{LPC:connect-to-open}) requires a facility to be open in order to connect clients. Constraint (\ref{LPC:capacity}) is the capacity constraint.   

It is well known that the basic LP has unbounded integrality gap, even if we are allowed to violate the cardinality constraint or the capacity constraint by a factor of $2-\epsilon$.  
\ifdefined\proceeding 
The description of the instance can be found in the full version of the paper. 
\else
In the gap instance for the capacity-violation setting, each facility has capacity $u$, $k$ is $2u-1$, and the metric consists of $u$ isolated groups each of which has $2$ facilities and $2u-1$ clients that are all collocated. In other words, the distances within a group are all 0 but the distances between groups are nonzero. Any integral solution for this instance has to have a group with at most one open facility. Therefore, even with $(2-2/u)$-capacity-violation, we have to connect 1 client in this group to open facilities in other groups. On the other hand, a fractional solution to the basic LP relaxation opens $2-1/u$ facilities in each group and serves the demand of each group using only the facilities in that group. Note that the gap instance disappears if we allow a capacity violation of 2.\footnote{A similar instance can be given to show that the gap is still unbounded when the cardinality constraint is violated, instead of the capacity constraint, by less than 2: let $k=u+1$ and each group have $2$ facilities and $u+1$ clients.} 

\fi
In order to overcome the gap in the cardinality-violation setting, Li \cite{Li16} introduced a novel LP for \CKM called the configuration LP, which we formally state below. Let us fix a set $B\subseteq F$ of facilities. Let $\ell = \Theta(1/\eps)$ and $\ell_1 = \Theta(\ell)$ be sufficiently large integers. Let $\calS = \set{S \subseteq B: |S| \leq \ell_1}$ and $\tcalS = \calS \cup \set{\bot}$, where $\bot$ stands for ``any subset of $B$ with size more than $\ell_1$''; for convenience, we also treat $\bot$ as a set such that $i \in \bot$ holds for every $i \in B$.   For $S \in \calS$, let $z^B_S$ indicate the event that the set of open facilities in $B$ is exactly $S$ and $z^B_\bot$ indicate the event that the number of open facilities in $B$ is more than $\ell_1$. 

For every $S \in \tcalS$ and $i \in S$, $z^B_{S, i}$ indicates the event that $z^B_S = 1$ and $i$ is open. (If $i \in B$ but $i \notin S$, then the event will not happen.)  Notice that when $i \in S \neq \bot$, we always have $z^B_{S, i} = z^B_S$; we keep both variables for notational purposes.  For every $S \in \tcalS, i \in S$ and client $j \in C$,  $z^B_{S, i, j}$ indicates the event that $z^B_{S, i} = 1$ and $j$ is connected to $i$.  In an integral solution, all the above variables are $\set{0,1}$ variables.  The following constraints are valid. To help understand the constraints, it is good to think of $z^B_{S, i}$ as $z^B_S \cdot y_i$ and $z^B_{S, i, j}$ as $z^B_S \cdot x_{i, j}$.


\noindent\hspace*{-0.05\textwidth}\begin{minipage}{0.56\textwidth}
\noindent\begin{alignat}{2}
\sum_{S \in \tcalS }z^B_S &= 1; \label{CLPC:add-to-one} \\
\sum_{S \in \tcalS:i \in S}z^B_{S, i} &= y_i, & \forall& i \in B; \label{CLPC:add-to-y} \\
\sum_{S \in \tcalS : i \in S}z^B_{S, i, j} &= x_{i,j}, & \forall& i \in B, j \in C; \label{CLPC:add-to-x} \\
0 \leq z^B_{S, i, j} \leq z^B_{S, i} &\leq z^B_S, & \forall & S \in \tcalS, i \in S, j \in C; \label{CLPC:non-negative}
\end{alignat}
\end{minipage}
\begin{minipage}{0.46\textwidth}
\vspace*{20pt}
\begin{alignat}{2}
z^B_{S, i}  &= z^B_S, &\forall& S \in \calS, i \in S; \label{CLPC:i-irrelevant} \\
\sum_{i \in S}z^B_{S, i, j} &\leq z^B_S, &\forall& S \in \tcalS, j \in C; \label{CLPC:j-connection-bound}\\
\sum_{j \in C}z^B_{S, i, j} &\leq u_i z^B_{S, i}, & \forall & S \in \tcalS, i \in S; \label{CLPC:capacity} \\
\sum_{i \in B}z^B_{\bot, i} &\geq \ell_1 z^B_\bot. \label{CLPC:more-than-ell-facilities} \\
\nonumber
\end{alignat}
\end{minipage}

Constraint~\eqref{CLPC:add-to-one} says that $z^B_S = 1$ for exactly one $S \in \tcalS$. Constraint~\eqref{CLPC:add-to-y} says that if $i$ is open then there is exactly one $S \in \tcalS$ with $z^B_{S, i} = 1$.  Constraint~\eqref{CLPC:add-to-x} says that if $j$ is connected to $i$ then there is exactly one $S \in \tcalS$ such that $z^B_{S, i, j} = 1$. Constraint~\eqref{CLPC:non-negative} is by the definition of variables. Constraint~\eqref{CLPC:i-irrelevant} holds as we mentioned earlier.  Constraint~\eqref{CLPC:j-connection-bound} says that if $z^B_S = 1$ then $j$ can be connected to at most 1 facility in $S$. Constraint~\eqref{CLPC:capacity} is the capacity constraint. Constraint~\eqref{CLPC:more-than-ell-facilities} says that if $z^B_{\bot} = 1$, there are at least $\ell_1$ open facilities in $B$.

The configuration LP is obtained from the basic LP by adding the $z$ variables and Constraints~\eqref{CLPC:add-to-one} to~\eqref{CLPC:more-than-ell-facilities} for every $B \subseteq F$. Since there are exponentially many subsets $B \subseteq F$, we don't know how to solve this LP efficiently. However, note that there are only polynomially many ($n^{O(\ell_1)}$) $z^B$ variables for a fixed $B \subseteq F$. Given a fractional solution $(x, y)$ to the basic LP relaxation, we can construct the values of $z^B$ variables and check their feasibility for Constraints~\eqref{CLPC:add-to-one} to~\eqref{CLPC:more-than-ell-facilities} in polynomial time as in \cite{Li16}. Our rounding algorithm either constructs an integral solution with the desired properties, or outputs a set $B \subseteq F$ such that Constraints~\eqref{CLPC:add-to-one} to~\eqref{CLPC:more-than-ell-facilities} are infeasible. In the latter case, we can find a constraint in the configuration LP that $(x, y)$ does not satisfy.  Then we can run the ellipsoid method and the rounding algorithm in an iterative way (see, e.g., \cite{CFL00, ASS14}).
\smallskip

\noindent {\bf Notations}\ \ From now on, we fix a solution $\left( \set{x_{i,j}:i\in F, j\in C}, \set{y_i:i\in F}\right)$ to the basic LP. We define $\dav(j) := \sum_{i \in F}x_{i, j}d(i, j)$ to be the connection cost of $j$, for every $j \in C$. Let $D_i := \sum_{j \in C}x_{i,j}\left( d(i,j) + \dav(j) \right)$ for every $i \in F$, and $D_S := \sum_{i \in S}D_i$ for every $S \subseteq F$. We denote the value of the solution $(x, y)$ by $\LP := \sum_{i \in F, j \in C}x_{i,j}d(i,j)=\sum_{j \in C}\dav(j)$. Note that $D_F = \sum_{i \in F,j \in C}x_{i,j}\left( d(i,j) + \dav(j) \right)= \sum_{i \in F, j\in C}x_{i, j}d(i, j)  + \sum_{j \in C}\dav(j)\sum_{i \in F}x_{i, j} =2\LP $. For any set $F' \subseteq F$ of facilities and $C' \subseteq C$ of clients, we shall let $x_{F', C'} := \sum_{i \in F', j \in C'} x_{i, j}$; we simply use $x_{i, C'}$ for $x_{\set{i}, C'}$ and  $x_{F', j}$ for $x_{F', \set{j}}$. For any $F' \subseteq F$, let $y_{F'} := \sum_{i \in F'}y_i$. Let $d(A, B):= \min_{i \in A, j \in B} d(i, j)$ denote the minimum distance between $A$ and $B$, for any $A, B \subseteq F \cup C$; we simply use $d(i, B)$ for $d(\set{i}, B)$.  \smallskip

\noindent {\bf Moving of Demands}\ \ After the set of open facilities is decided, the optimum connection assignment from clients to facilities can be computed by solving the minimum cost $b$-matching problem. Due to the integrality of the matching polytope, we may allow the connections to be fractional. That is, if there is a good fractional assignment, then there is a good integral assignment. So we can use the following framework to design and analyze the rounding algorithm. Initially there is one unit of demand at each client $j \in C$. During the course of our algorithm, we move demands fractionally within $F \cup C$; moving $\alpha$ units of demand from $i$ to $j$ incurs a cost of $\alpha d(i, j)$. At the end, all the demands are moved to $F$ and each facility $i\in F$ has at most $(1+O(\frac{1}{\ell}))u_i$ units of demand. We open a facility if it has positive amount of demand. Our goal is to bound the total moving cost by $O(\ell^5)\LP$ and the number of open facilities by $k$.

%% file: setup.tex
\section{Representatives, Black Components, and Groups}\label{section:setup}

Our algorithm starts with bundling facilities together with a three-phase process each of which creates bigger and bigger clusters. At the end, we have a nicely formed network of sufficiently big clusters of facilities. See Figure~\ref{fig:setup} for illustration of the three-phase clustering.  

\begin{figure}
	\centering
	\includegraphics[width=0.95\textwidth]{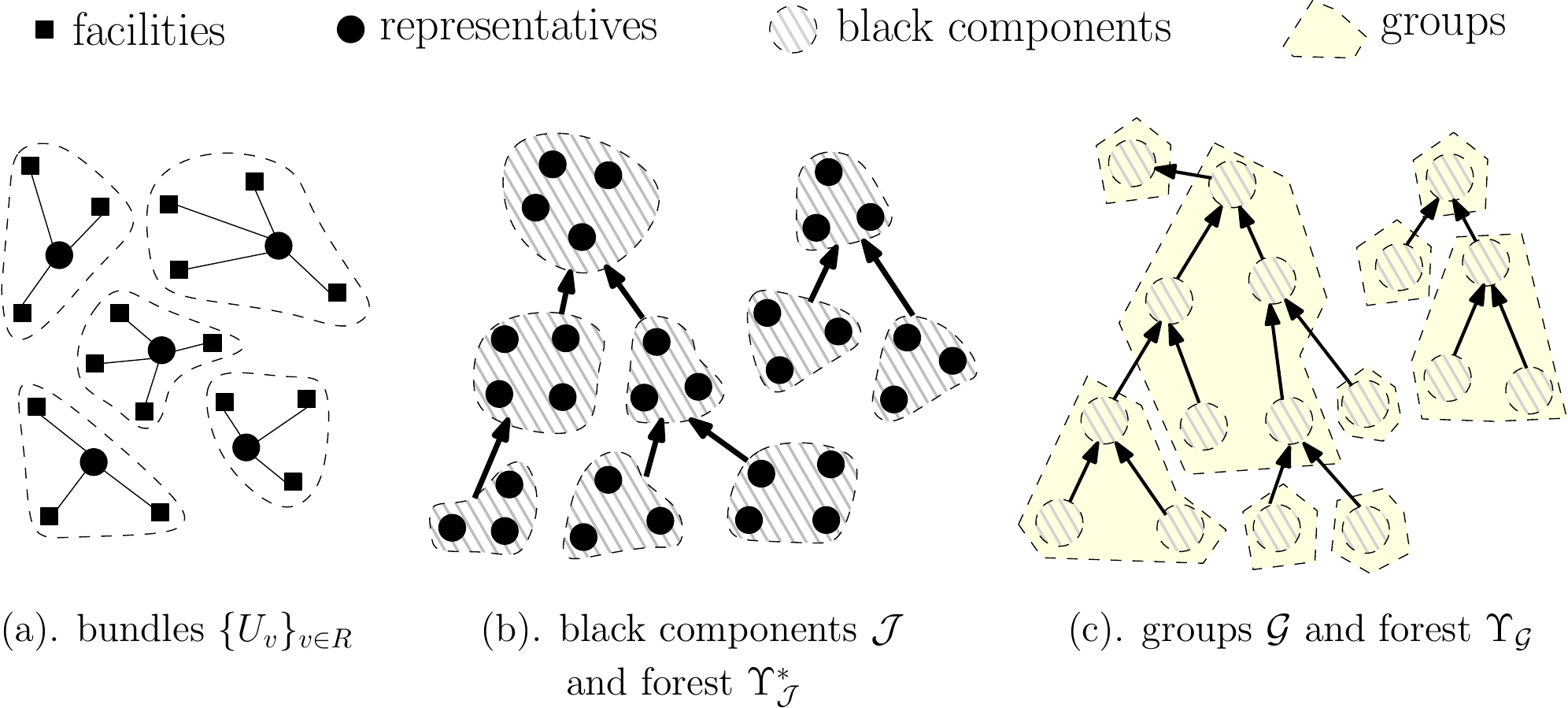}
	\caption{The three-phase clustering procedure. In the first phase (Figure (a)), we partition $F$ into bundles, centered at the set $R$ of representatives. In the second phase (Figure (b)), we partition $R$ into a family $\calJ$ of black components and construct a degree-2 rooted forest over $\calJ$. In the third phase (Figure(c)), we partition $\calJ$ into a family $\calG$ of groups; $\Upsilon_\calG$ is formed from $\Upsilon^*_\calJ$ by contracting each group into a single node.}
	\label{fig:setup}
\end{figure}


\subsection{Representatives, Bundles and Initial Moving of Demands}

In the first phase, we use a standard approach to facility location problems (\cite{LV92a, STA97, CGT99, Li16}) to partition the facilities into \emph{bundles} $\{U_v\}_{v \in R}$, where each bundle $U_v$ is associated with a center $v \in C$ that is called a \emph{representative} and $R \subseteq C$ is the set of representatives. Each bundle $U_v$ has a total opening at least $1/2$.  

Let $R = \emptyset$ initially. Repeat the following process until $C$ becomes empty: we select the client $v \in C$ with the smallest $\dav(v)$ and add it to $R$; then we remove all clients $j$ such that $d(j, v) \leq 4 \dav(j)$ from $C$ (thus, $v$ itself is removed). We use $v$ and its variants to index representatives, and $j$ and its variants to index general clients.  The family $\set{U_v:v \in R}$ is the Voronoi diagram of $F$ with $R$ being the centers: let $U_v = \emptyset$ for every $v \in R$ initially; for each location $i \in F$, we add $i$ to $U_v$ for $v \in R$ that is closest to $i$. For any subset $V \subseteq R$, we use $U(V) := \union_{v \in V} U_v$ to denote the union of Voronoi regions with centers $V$.

\begin{lemma} \label{claim:representatives}
The following statements hold:
	\begin{properties}{claim:representatives}
		\item \label{property:representatives-far-away} for all $v, v' \in R, v \neq v'$, we have $d(v, v') > 4 \max\set{\dav(v), \dav(v')}$ 
		\item \label{property:near-a-representative} for all $j \in C$, there exists $v \in R$, such that $\dav(v) \leq \dav(j)$ and $d(v, j) \leq 4 \dav(j)$; 
		\item \label{property:bundle-large} $y_{U_v} \geq 1/2$ for every $v \in R$;  
		\item \label{property:facility-to-representative} for any $v \in R$, $i \in U_v$, and $j \in C$, we have $d(i, v) \leq d(i, j) + 4 \dav(j)$.
	\end{properties}
\end{lemma}

\ifdefined\proceeding
\else
\begin{proof}
	First consider Property~\ref{property:representatives-far-away}.  Assume $\dav(v) \leq \dav(v')$. When we add $v$ to $R$, we remove all clients $j$ satisfying $d(v, j) \leq 4 \dav(j)$ from $C$. If $v' \in R$, then it must have been $d(v, v') > 4 \dav(v')$. For Property~\ref{property:near-a-representative}, just consider the iteration in which $j$ is removed from $C$.  The representative $v$ added to $R$ in this iteration satisfy the property. Then consider Property~\ref{property:bundle-large}. By Property~\ref{property:representatives-far-away}, we have $B:=\set{i \in F:d(i, v) \leq 2 \dav(v)} \subseteq U_v$. Since $\dav(v)=\sum_{i \in F}x_{i,v}d(i,v)$ and $\sum_{i \in F}x_{i,v} = 1$, we have $\dav(v) \geq (1-x_{B, v})2 \dav(v)$, implying $y_{U_v} \geq y_B \geq x_{B, v} \geq 1/2$, due to Constraint~\eqref{LPC:connect-to-open}.
	
	Then we consider Property~\ref{property:facility-to-representative}. By Property~\ref{property:near-a-representative}, there is a client $v' \in R$ such that $\dav(v') \leq \dav(j)$  and $d(v', j) \leq 4 \dav(j)$. Since $d(i, v) \leq d(i, v')$ as $v' \in R$ and $i$ was added to $U_v$, we have $d(i, v) \leq d(i, v') \leq d(i, j) + d(j, v') \leq d(i, j) + 4 \dav(j)$.
\end{proof}
\fi

The next lemma shows that moving demands from facilities to their corresponding representative doesn't cost much.
\begin{lemma} \label{lemma:moving-to-representatives}
For every $v\in R$,  we have 
$\sum_{i \in U_v}x_{i,C}d(i,v) \leq O(1) D_{U_v}$.
\end{lemma}
\ifdefined\proceeding
\else
\begin{proof}
	By Property~\ref{property:facility-to-representative}, we have $d(i, v) \leq d(i, j) + 4\dav(j)$ for every $i \in U_v$ and $j \in C$.  Thus, 
	$$\sum_{i \in U_v}x_{i,C} d(i,v) \leq \sum_{i \in U_v, j \in C}x_{i,j} \big(d(i, j) + 4\dav(j)\big) \leq \sum_{i \in U_v}4D_i = 4D_{U_v}. $$ 
\end{proof}
\fi

Since $\set{U_v:v\in R}$ forms a partition of $F$, we get the following corollary.
\begin{corollary} \label{cor:moving-to-representatives}
$\sum_{v\in R,i \in U_v}x_{i,C}d(i,v) \leq O(1) \LP$.
\end{corollary}

\noindent{\bf Initial Moving of Demands}\ \ With this corollary, we now move all the demands from $C$ to $R$. First for every $j \in C$ and $i \in F$, we move $x_{i, j}$ units of demand from $j$ to $i$. The moving cost of this step is exactly $\LP$.  After the step, all demands are at $F$ and every $i \in F$ has $x_{i, C}$ units of demand. Then, for every $v \in R$ and $i \in U_v$, we move the $x_{i,C}$ units of demand at $i$ to $v$. The moving cost for this step is $O(1)\LP$. Thus, after the initial moving, all demands are at the set $R$ of representatives: a representative $v$ has $x_{U_v, C}$ units of demand. 

\subsection{Black Components}

	In the second phase, we employ the minimum-spanning-tree construction of \cite{Li16} to partition the set $R$ of representatives into a family $\calJ$ of so-called \emph{black components}. There is a degree-2 rooted forest $\Upsilon^*_\calJ$ over $\calJ$ with many good properties.  For example, each non-root black component is not far away from its parent, and each root black component of $\Upsilon^*_\calJ$ contains a total opening of $\Omega(\ell)$. (For simplicity, we say the total opening at a representative $v \in R$ is $y_{U_v}$, which is the total opening at the bundle $U_v$.)  The forest in \cite{Li16} can have a large degree, while our algorithm requires the forest to have degree 2. This property is guaranteed by using the left-child-right-sibling representation.
	
\ifdefined\proceeding
	Due to the page limit, we leave the description of the framework of \cite{Li16} to the full version of the paper, and give its summary in the following lemma:
\else
	We now describe the framework of \cite{Li16}. We run the classic Kruskal's algorithm to find the minimum spanning tree $\MST$ of the metric $(R, d)$, and then color the edges in $\MST$ in black, grey or white. In Kruskal's algorithm, we maintain the set $E_\MST$ of edges added to $\MST$ so far and a partition $\mathfrak{P}$ of $R$. Initially, we have $E_\MST = \emptyset$ and $\mathfrak{P} = \set{\{v\}: v\in R}$. The length of an edge $e \in {R \choose 2}$ is the distance between the two endpoints of $e$. We sort all edges in $R \choose 2$ in the ascending order of their lengths, breaking ties arbitrarily. For each pair $({v}, {v'})$ in this order, if ${v}$ and ${v'}$ are not in the same partition in $\mathfrak{P}$, we add the edge $({v}, {v'})$ to $E_\MST$ and merge the two partitions containing ${v}$ and ${v'}$ respectively.  
	
	We then color edges in $E_\MST$.  For every $v \in R$, we say the \textit{weight} of $v$ is $y_{U_v}$; so every representative $v \in R$ has weight at least $1/2$ by Property~\ref{property:bundle-large}. For a subset $J \subseteq R$ of representatives, we say $J$ is big if the weight of $J$ is at least $\ell$, i.e, $y_{U(J)} \geq \ell$; we say $J$ is small otherwise.  For any edge $e = ({v}, {v'})\in E_\MST$, we consider the iteration in Kruskal's algorithm in which the edge $e$ is added to $\MST$. After the iteration we merged the partition $J_{v}$ containing ${v}$ and the partition $J_{v'}$ containing ${v'}$ into a new partition $J_{v} \cup J_{v'}$.  If both $J_{v}$ and $J_{v'}$ are small, then we call $e$ a black edge.  If $J_{v}$ is small and $J_{v'}$ is big, we call $e$ a grey edge, directed from ${v}$ to ${v'}$; similarly, if $J_{v'}$ is small and $J_{v}$ is big, $e$ is a grey edge directed from ${v'}$ to ${v}$. If both $J_{v}$ and $J_{v'}$ are big, we say $e$ is a white edge. So, we treat black and white edges as undirected edges and grey edges as directed edges. 
	
	We define a black component of $\MST$ to be a maximal set of vertices connected by black edges. Let $\calJ$ be the set of all black components.  Thus $\calJ$ indeed forms a partition of $R$. We contract all the black edges in $\MST$ and remove all the white edges. The resulting graph is a forest $\Upsilon_\calJ$ of trees over black components in $\calJ$. Each edge is a directed grey edge. Later in Lemma~\ref{lemma:contracted-tree}, we show that the grey edges are directed towards the roots of the trees. For every component $J \in \calJ$, we define $L(J) := d(J, R \setminus J)$ to be the shortest distance between any representative in $J$ and any representative not in $J$.  
	
	A component $J$ in the forest $\Upsilon_\calJ$ may have many child-components. To make the forest binary, we use the left-child-right-sibling binary-tree representation of trees.  To be more specific, for every component $J'$, we sort all its child-components $J$ according to non-decreasing order of $L(J)$.  We add a directed edge from the first child to $J'$ and a directed edge between every two adjacent children in the ordering, from the child appearing later in the ordering to the child appearing earlier.  Let $\Upsilon^*_\calJ$ be the new forest. $\Upsilon^*_\calJ$ naturally defines a new child-parent relationship between components.
\fi

	\begin{lemma}\label{lemma:contracted-tree}
\ifdefined\proceeding
		There is an efficient algorithm to partition $R$ into a set $\calJ$ of \emph{black components} (or components, for simplicity) and construct a rooted forest $\Upsilon^*_\calJ$ over $\calJ$, such that if we let $L(J) = d(J, R \setminus J)$ for every black component $J \in \calJ$, then the following properties hold:
\else
		$\calJ$ and $\Upsilon^*_\calJ$ satisfy the following properties:
\fi
		\begin{properties}{lemma:contracted-tree}
			\item for every $J \in \calJ$, there is a spanning tree over the representatives in $J$ such that for every edge $(v, v')$ in the spanning tree we have $d(v, v')\leq L(J)$; \label{property:black-edges-short}
			\item every root component $J \in \calJ$ of $\Upsilon^*_\calJ$ has $y_{U(J)} \geq \ell$ and every non-root component $J \in \calJ$ has $y_{U(J)} < \ell$; \label{property:black-weights-of-components}
			\item every root component $J \in \calJ$ of $\Upsilon^*_\calJ$ has either $y_{U(J)} < 2\ell$ or $|J| = 1$; \label{property:black-root-is-big}
			\item for any non-root component $J$ and its parent $J'$, we have $L(J) \geq L(J')$; \label{property:black-lengths-decrease}
			\item for any non-root component $J$ and its parent $J'$, we have $d(J, J') \leq O(\ell)L(J)$; \label{property:black-grey-edge-length}
			\item every component $J$ has at most two children. \label{property:black-binary}
		\end{properties}
	\end{lemma}

\ifdefined\proceeding
\else
	The rest of the section is dedicated to the proof of Lemma~\ref{lemma:contracted-tree}. We first prove some of the above properties for \emph{the original forest} $\Upsilon_\calJ$. We show that all black edges between the representatives in $J$ are considered before all the edges in $J \times (R \setminus J)$ in Kruskal's algorithm.  Assume otherwise.  Consider the first edge $e$ in $J \times (R \setminus J)$ we considered. Before this iteration, $J$ is not connected yet.   Then we add $e$ to the minimum spanning tree; since $J$ is a black component, $e$ is gray or white. In either case, the new partition $J'$ formed by adding $e$ will have weight more than $\ell$. This implies all edges in $J' \times (R \setminus J')$ added later to the MST are not black. Moreover, $J\setminus J', J' \setminus J$ and $J \cap J'$ are all non-empty. This contradicts the fact that $J$ is a black component.   Therefore, all black edges in $J$ has length at most $L(J)$, implying Property~\ref{property:black-edges-short} .
	
	Focus on a tree $T$ in the initial forest $\Upsilon_\calJ$ and any small black component $J$ in $T$. All black edges between the representatives in $J$ are added to $\MST$ before any edge in $J \times (R \setminus J)$. The first edge in $J \times (R \setminus J)$ added to $\MST$ is a grey edge directed from $J$ to some other black component: it is not white because $J$ is small; it is not black since $J$ is a black component. Thus, it is a grey edge in $T$.   Therefore, the growth of the tree $T$ in Kruskal's algorithm is as follows. The first grey edge in $T$ is added between two black components, one of them is big and the other is small.  We define the root of $T$ to be the big component.  At each time, we add a new small black component  $J$ to the current $T$ via a grey edge directed from $J$ to $T$. (During this process, white edges incident to $T$ may be added.)  So, the tree $T$ is a rooted tree with grey edges, where all edges are directed towards the root. So, Property~\ref{property:black-weights-of-components} holds.    Moreover, the length of the grey edge between $J$ and its parent $J'$ is $d(J, J') = L(J)$, which is stronger than Property~\ref{property:black-grey-edge-length}. Since $d(J, J') \geq d(J', R \setminus J') = L(J')$, we have Property~\ref{property:black-lengths-decrease}.
	
	The root $J$ of $T$ is a big black component. Suppose it contains two or more representatives; so it's not a singleton. Consider the last black edge $({v}, {v'})$ added between $J_{v}$ and $J_{v'}$ to make $J = J_{v} \cup J_{v'}$. Since $({v},{v'})$ is a black edge, both $J_{v}$ and $J_{v'}$ are small, i.e. $y_{U(J_{v})}, y_{U(J_{v'})} <  \ell$. Therefore, we have $y_{U(J)} = y_{U(J_{v})} +y_{U(J_{v'})} < 2 \ell$, proving Property \ref{property:black-root-is-big}. 	
	
	Now, we move on to prove all the properties of the lemma for \emph{the final forest} $\Upsilon^*_{\calJ}$.  We used the left-child-right-sibling binary-tree representation of $\Upsilon_\calJ$ to obtain $\Upsilon^*_\calJ$. Thus , Property~\ref{property:black-binary} holds for $\Upsilon^*_\calJ$.  Property~\ref{property:black-edges-short} is independent of the forest and thus still holds for $\Upsilon^*_\calJ$. A component is a root in $\Upsilon_\calJ$ if and only if it is a root in $\Upsilon^*_\calJ$. Thus, properties~\ref{property:black-weights-of-components} and \ref{property:black-root-is-big} are maintained for $\Upsilon^*_\calJ$.   Since we sorted the children of a component according to $L$ values before constructing the left-child-right-sibling binary tree, Property~\ref{property:black-lengths-decrease} holds for $\Upsilon^*_\calJ$. 
	
	For every component $J$ and its parent $J'$ in the forest $\Upsilon^*_{\calJ}$, we have $L(J) = d(J, R \setminus J) = d(J, J'')$, where $J''$ is the parent of $J$ in the initial forest $\Upsilon_\calJ$. $J'$ is either $J''$, or a child of $J''$ in $\Upsilon_\calJ$.  In the former case, we have $d(J, J') = d(J, J'') = L(J)$.  In the latter case, we have that $d(J', J'') = L(J') \leq L(J) = d(J, J'')$. Due to Property~\ref{property:black-edges-short}, we have a path connecting some representative in $J$ to some representative in $J'$, with internal vertices being representatives in $J''$, and all edges having length at most $L(J)$. Moreover, there are at most $4\ell$ representatives in $J''$ due to Properties~\ref{property:black-weights-of-components}, \ref{property:black-root-is-big}, and \ref{property:bundle-large}. Thus, we have $d(J, J') \leq O(\ell)d(J, J'') = O(\ell)L(J)$.  Thus, Property~\ref{property:black-grey-edge-length} holds for $\Upsilon^*_\calJ$.   This finishes the proof of Lemma~\ref{lemma:contracted-tree}.
\fi
\subsection{Groups}

In the third phase, we apply a simple greedy algorithm to the forest $\Upsilon^*_{\calJ}$ to partition the set $\calJ$ of black components into a family $\calG$ of \emph{groups}, where each group $G \in \calG$ contains many black components that are connected in $\Upsilon^*_\calJ$.  By contracting each group $G \in \calG$, the forest $\Upsilon^*_\calJ$ over the set $\calJ$ of black components becomes a forest $\Upsilon_\calG$ over the set $\calG$ of groups.  Each group has a total opening of $\Omega(\ell)$, unless it is a leaf-group in $\Upsilon_\calG$.

We partition the set $\calJ$ into groups using a technique similar to \cite{BFR15, BRU15}. 
For each rooted tree $T = (\calJ_T, E_T)$ in $\Upsilon^*_\calJ$, we construct a group $G$ of black components as follows. Initially, let $G$ contain the root component of $T$.  While $\sum_{J \in G} y_{U(J)} < \ell$ and $G \neq \calJ_T$, repeat the following procedure. Choose the component $J \in \calJ_T \setminus G$ that is adjacent to $G$ in $T$, with the smallest $L$-value, and add $J$ to $G$.
	
	Thus, by the construction $G$ is connected in $T$. After we have constructed the group $G$, we add $G$ to $\calG$. We remove all black components in $G$ from $T$. Then, each $T$ is broken into many rooted trees; we apply the above procedure recursively for each rooted tree. 
	
	So, we have constructed a partition $\calG$ for the set $\calJ$ of components. If for every $G \in \calG$, we contract all components in $G$ into a single node, then the rooted forest $\Upsilon^*_\calJ$ over $\calJ$ becomes a rooted forest $\Upsilon_\calG$ over the set $\calG$ of groups.  $\Upsilon_\calG$ naturally defines a parent-child relationship over $\calG$.  The following lemma uses Properties~\ref{property:black-edges-short} to \ref{property:black-binary} of $\calJ$ and the way we construct $\calG$.

\begin{lemma} \label{lemma:groups}
	The following statements hold for the set $\calG$ of groups and the rooted forest $\Upsilon_\calG$ over $\calG$:
		\begin{properties}{lemma:groups}
			\item any root group $G \in \calG$ contains a single root component $J \in \calJ$; \label{property:groups-root-sington}
			\item if $G \in \calG$ is not a root group,  then $\sum_{J \in G}y_{U(J)} < 2\ell$; \label{property:groups-non-root-small}
			\item if $G \in \calG$ is a non-leaf group, then $\sum_{J \in G} y_{U(J)} \geq \ell$;
			\label{property:groups-non-leaf-big}
			\item let $G \in \calG, G' \in \calG$ be the parent of $G$, $J \in G$ and $v \in J$, then the distance between $v$ and any representative in $\union_{J' \in G'}J'$ is at most $O(\ell^2)L(J)$; \label{property:groups-close}
			\item any group $G$ has at most $O(\ell)$ children. \label{property:groups-few-children}
		\end{properties}
\end{lemma}	

\ifdefined\proceeding
\else
\begin{proof}
	For a root component $J$, we have $y_{U(J)} \geq \ell$ by Property~\ref{property:black-weights-of-components}. Thus, any root group $G$ contains a single root component $J$, which is exactly Property~\ref{property:groups-root-sington}.
	
	When constructing the group $G$ from the tree $T = (\calJ_T, E_T)$, the terminating condition is $G = \calJ_T$ or $\sum_{J \in G}y_{U(J)} \geq \ell$. Thus, if $G$ is not a leaf-group, then the condition $G = \calJ_T$ does not hold; thus we have $\sum_{J \in G}y_{U(J)} \geq \ell$, implying Property \ref{property:groups-non-leaf-big}. 
	
	By Property~\ref{property:black-weights-of-components}, any non-root component $J$ has $y_{U(J)} < \ell$. Thus, if $G$ is not a root group, the terminating condition constructing $G$ implies that $G$ had total weight less than $\ell$ right before the last black component was added to it. Then we have $\sum_{J \in G}y_{U(J)} < 2\ell$, implying Property~\ref{property:groups-non-root-small}.
	
	Now, consider Property~\ref{property:groups-close}. From Property~\ref{property:black-lengths-decrease}, it is easy to see that the group $G$ constructed from the tree $T = (\calJ_T, E_T)$ has the following property: the $L$ value of any component in $G$ is at most the $L$-value of any component in $\calJ_T \setminus G$.   Let $G$ be a non-root group and $G'$ be its parent; let $J \in G$ and $J' \in G'$ be black components. Thus, there is a path in $\Upsilon^*_\calJ$ from $J$ to $J'$, where components have $L$-values at most $L(J)$.   The edges in the path have length at most $O(\ell)L(J)$ by Property~\ref{property:black-grey-edge-length}.  Moreover, Property~\ref{property:black-edges-short} implies that the representatives in each component in the path are connected by edges of length at most $L(J)$. Thus,  we can find a path from $v$ to $v'$ that go through representatives in $\union_{J'' \in G \cup G'} J''$, and every edge in the path has length at most $O(\ell)L(J) = O(\ell)d(J, R\setminus J)$.  By Property~\ref{property:bundle-large}, \ref{property:black-weights-of-components} and \ref{property:black-root-is-big}, the total representatives in the components contained in $G$ (as well as in $G'$) is at most $4\ell$.  Thus, the distance between $v$ and $v'$ is at most $O(\ell^2)L(J)$, which is exactly Property~\ref{property:groups-close}.	
	
	Finally, since the forest $\Upsilon^*_\calJ$ is binary and every group $G \in \calG$ contains at most $O(\ell)$ components, we have that every group $G$ contains at most $O(\ell)$ children, implying Property~\ref{property:groups-few-children}.
\end{proof}
\fi

%% file: distribution.tex
\section{Constructing Local Solutions}
	\label{sec:distribution}
	
	In this section, we shall construct a local solution, or a distribution of local solutions, for a given set $V \subseteq R$ which is the union of some black components.  A local solution for $V$ contains a pair $(S \subseteq U(V), \beta \in \R_{\geq 0}^{U(V)})$, where $S$ is the facilities we open in $U(V)$ and $\beta_i$ for each $i \in U(V)$ is the amount of supply at $i$: the demand that can be satisfied by $i$. Thus $\beta_i = 0$ if $i \in U(V) \setminus S$. We shall use the supplies at $U(V)$ to satisfy the $x_{U(V), C}$ demands at $V$ after the initial moving of demands; thus, we require $\sum_{i \in U(V)}\beta_i = x_{U(V), C}$. There are two other main properties we need the distribution to satisfy: (a) the expected size of $S$ from the distribution is not too big, and (b) the cost of matching the demands at $V$ and the supplies at $U(V)$ is small. 
	
	We distinguish between \emph{concentrated} black components and \emph{non-concentrated} black components.  Roughly speaking, a component $J \in \calJ$ is concentrated if in the fractional solution $(x, y)$, for most clients $j \in C$, $j$ is either almost fully served by facilities in $U(J)$, or almost fully served by facilities in $F \setminus U(J)$.  We shall construct a distribution of local solutions for each concentrated component $J$. We require Constraints~\eqref{CLPC:add-to-one} to~\eqref{CLPC:more-than-ell-facilities} to be satisfied for $B = U(J)$ (if not, we return the set $U(J)$ to the separation oracle) and let $z^B$ be the vector satisfying the constraints. Roughly speaking, the $z^B$-vector defines a distribution of local solutions for $V$.  A local solution $(S, \beta)$ is good if $S$ is not too big and the total demand $\sum_{i \in S}\beta_i$ satisfied by $S$ is not too small.  Then, our algorithm randomly selects $(S, \beta)$ from the distribution defined by $z^B$, under the condition that $(S, \beta)$ is good.  The fact that $J$ is concentrated guarantees that the total mass of good local solutions in the distribution is large; therefore the factors we lose due to the conditioning are small. 
	
	For non-concentrated components,  we construct a single local solution $(S, \beta)$, instead of a distribution of local solutions. Moreover, the construction is for the union $V$ of some non-concentrated components, instead of an individual component.  The components that comprise $V$ are close to each other; by the fact that they are non-concentrated, we can move demands arbitrarily within $V$, without incurring too much cost.  Thus we can essentially treat the distances between representatives in $V$ as $0$. Then we are only concerned with two parameters for each facility $i \in U(V)$: the distance from $i$ to $V$ and the capacity $u_i$.  Using a simple argument, the optimum fractional local solution (that minimizes the cost of matching the demands and supplies) is almost integral: it contains at most 2 fractionally open facilities.  By fully opening the two fractional facilities, we find an integral local solution with small number of open facilities. 

	The remaining part of this section is organized as follows.  We first formally define concentrated black components, and explain the importance of the definition. We then define the earth-mover-distance, which will be used to measure the cost of satisfying demands using supplies. The construction of local solutions for concentrated components and non-concentrated components will be stated in Theorem~\ref{thm:concentrated-sets} and Lemma~\ref{lemma:non-concentrated-sets} respectively. 
\ifdefined\proceeding
	Due to the page limit, their proofs will only appear in the full version of the paper. 
\smallskip

\noindent{\bf Concentrated Black Components}\ \ 
\else
	\subsection{Concentrated Black Components and Earth Mover Distance}
\fi	
	The definition of concentrated black component is the same as that of \cite{Li16}, except that we choose the parameter $\ell_2$ differently. 		
	\begin{definition}\label{def:concentrated-compoenents}
		Define $\pi_J = \sum_{j \in C}x_{U(J),j} (1-x_{U(J), j})$, for every black component $J \in \calJ$.   A black component $J \in \calJ$ is said to be \emph{concentrated} if $\pi_J \leq x_{U(J), C} /\ell_2$, and \emph{non-concentrated} otherwise, where $\ell_2 = \Theta(\ell^3)$ is large enough.
	\end{definition}
	
	We use $\calJ^\sfC$ to denote the set of concentrated components and $\calJ^\sfN$ to denote the set of non-concentrated components. The next lemma from \cite{Li16} shows the importance of $\pi_J$.
\ifdefined\proceeding
\else 
	For the completeness of the paper, we include its proof here.
\fi
	\begin{lemma}
		\label{lemma:x-j-times-one-minus-x-j-small}
		For any $J \in \calJ$, we have  $L(J) \pi_J \leq O(1)D_{U(J)}$.
	\end{lemma}

\ifdefined\proceeding
\else
	\begin{proof}
		Let $B = U(J)$. For every $i \in B, j \in C$, we have $d(i, J) \leq d(i, j) + 4\dav(j)$ by Property~\ref{property:facility-to-representative} and the fact that $i \in U_v$ for some $v \in J$. Thus, 
		\begin{align*}
			&\quad L(J)\pi(J) = L(J)\sum_{j \in C}x_{B, j}(1-x_{B,j}) 
			= L(J)\sum_{j \in C, i \in B, i' \in F \setminus B}x_{i,j}x_{i',j}\\
			&\leq \sum_{i \in B, j \in C, i' \in F \setminus B}x_{i,j} x_{i',j} \cdot 2d(i', J) 
			\leq 2\sum_{i \in B, j \in C}x_{i,j}\sum_{i' \in F}x_{i',j}\big(d(i',j) + d(j, i) + d(i, J)\big)\\
			&=2\sum_{i \in B, j \in C}x_{i,j} \Big(\dav(j) + d(j, i) + d(i, J)\Big)
			\leq 2\sum_{i \in B, j \in C}x_{i,j} \Big(2d(i, j) + 5\dav(j) \Big) \\
			&=2\sum_{i \in B}(5D_i) = 10D_B.
		\end{align*}
		The first inequality is by $L(J) \leq 2 d(i', J)$ for any $i' \in F\setminus B = U_{R \setminus J}$: $d(i', R \setminus J) \leq d(i', J)$ implies $L(J) = d(R \setminus J, J) \leq d(R \setminus J, i') + d(i', J) \leq 2d(i', J)$. The second inequality is by triangle inequality and the third one is by $d(i, J) \leq d(i, j) + 4\dav(j)$. All the equalities are by simple manipulations of notations. 
	\end{proof}
\fi	
	
	Recall that $L(J) = d(J, R \setminus J)$ and $x_{U(J),C}$ is the total demand in $J$ after the initial moving. Thus, according to Lemma~\ref{lemma:x-j-times-one-minus-x-j-small}, if $J$ is not concentrated, we can use $D_{U(J)}$ to charge the cost for moving all the $x_{U(J), C}$ units of demand out of $J$, provided that the moving distance is not too big compared to $L(J)$. This gives us freedom for handling non-concentrated components. If $J$ is concentrated, the amount of demand that is moved out of $J$ must be comparable to $\pi_J$; this will be guaranteed by the configuration LP. 
	
\ifdefined\proceeding
	\smallskip
	
	\noindent{\bf Earth Mover Distance}\ \ 
\fi
	In order to measure the moving cost of satisfying demands using supplies, we define the earth mover distance:
	\begin{definition}[Earth Mover Distance] Given a set $V \subseteq R$ with $B = U(V)$, a demand vector $\alpha \in \R_{\geq 0}^V$ and a supply vector $\beta \in \R_{\geq 0}^B$ such that $\sum_{v \in V}\alpha_v \leq \sum_{i \in B}\beta_i$, the earth mover distance from $\alpha$ to $\beta$ is defined as
			$\EMD_V(\alpha, \beta) := \inf_{f}\sum_{v \in V, i \in B}f(v, i) d(v, i)$,
		where $f$ is over all functions from $V \times B$ to $\R_{\geq 0}$ such that 
		\vspace*{-8pt}
		
		\begin{itemize}
			\item $\sum_{i \in B}f(v, i) = \alpha_v$ for every $v \in V$;
			\item $\sum_{v \in V}f(v, i) \leq \beta_i$ for every $i \in B$.
		\end{itemize}
	\end{definition}
	For some technical reason, we allow some fraction of a supply to be unmatched. From now on, we shall use $\alpha_v = x_{U_v, C}$ to denote the amount of demand at $v$ after the initial moving.  For any set $V \subseteq R$ of representatives, we use $\alpha|_V$ to denote the vector $\alpha$ restricted to the coordinates in $V$.  

\ifdefined\proceeding
	We now summarize our constructions of local solutions for concentrated and non-concentrated black components, respectively.  
\else
	\subsection{Distributions of Local Solutions for Concentrated Components}
	\label{subsec:distri-con}
	In this section, we construct distributions for components in $\calJ^\sfC$, by proving:
\fi	
	\begin{theorem} \label{thm:concentrated-sets} 
		Let $J \in \calJ^\sfC$ and let $B = U(J)$.  Assume Constraints~\eqref{CLPC:add-to-one} to~\eqref{CLPC:more-than-ell-facilities} are satisfied for $B$. Then, we can find a distribution $(\phi_{S,\beta})_{S \subseteq B, \beta \in \R_{\geq 0}^B}$ of pairs $(S, \beta)$,  such that
		\begin{properties}{thm:concentrated-sets}
			\item $s_\phi:= \E_{(S, \beta) \sim \phi}|S| \in [y_B, y_B(1 + 2\ell\pi_J/x_{B,C})]$, and $s_\phi = y_B$ if $y_B > 2\ell$, \label{property:concentrated-S-expectation}
		\end{properties}
		and	for every $(S, \beta)$ in the support of $\phi$, we have
		\begin{properties}[2]{thm:concentrated-sets}
			\item $|S| \in \{\floor{s_\phi}, \ceil{s_\phi}\}$;   \label{property:concentrated-S}
			\item $\beta_i \leq (1+O(1/\ell))u_i$ if $i \in S$ and $\beta_i = 0$ if $i \in B \setminus S$; \label{property:concentrated-beta}
			\item $\sum_{i \in S} \beta_i = x_{B, C} = \sum_{v \in J}\alpha_v$. \label{property:concentrated-enough-demand}
		\end{properties}
		Moreover, the distribution $\phi$ satisfies
		\begin{properties}[5]{thm:concentrated-sets}
			\item the support of $\phi$ has size at most $n^{O(\ell)}$; \label{property:concentrated-support-small}
			\item $\E_{(S, \beta) \sim \phi}\EMD_J(\alpha|_J, \beta) \leq O(\ell^4) D_B$. \label{property:concentrated-cost}
		\end{properties}
	\end{theorem}
	
\ifdefined\proceeding
\else
	To prove the theorem, we first construct a distribution $\psi$ that satisfies most of the properties; then we modify it to obtain the final distribution $\phi$. Notice that a typical black component $J$ has $y_B \leq 2\ell$; however, when $J$ is a root component containing a single representative, $y_B$ might be very large. For now, let us just assume $y_B \leq 2\ell$.  We deal with the case where $|J| = 1$ and $y_B > 2\ell$ at the end of this section.
	
	Since Constraints~\eqref{CLPC:add-to-one} to~\eqref{CLPC:more-than-ell-facilities} are satisfied for $B$, we can use the $z^{B}$ variables satisfying these constraints to construct a distribution $\zeta$ over pairs $(\chi \in [0, 1]^{B \times C}, \mu \in [0, 1]^B)$, where $\mu$ indicates the set of open facilities in $B$ and $\chi$ indicates how the clients in $J$ are connected to facilities in $B$. Let $\calS = \{S \subseteq B: |S| \leq \ell_1\}$ and $\tcalS = S \cup \{\bot\}$ as in Section~\ref{sec:config-LP}. For simplicity, for any $\mu \in [0, 1]^B$, we shall use $\mu_B$ to denote $\sum_{i \in B}\mu_i$. For any $\chi \in [0, 1]^{B \times C}, i \in B$ and $j \in C$, we shall use $\chi_{i, C}$ to denote $\sum_{j \in C}\chi_{i, j}$, $\chi_{B, j}$ to denote $\sum_{i \in B}\chi_{i, j}$, and $\chi_{B, C}$ to denote $\sum_{i \in B} \chi_{i, C} = \sum_{j \in C}\chi_{B, j}$.
	
	The distribution $\zeta$ is defined as follows. Initially, let $\zeta_{\chi, \mu} = 0$ for all $\chi \in [0, 1]^{B \times C}$ and $\mu \in [0, 1]^B$. For each $S \in \tcalS$ such that $z^B_S > 0$, increase $\zeta_{\chi, \mu}$ by $z^B_S$ for the $\chi, \mu$ satisfying $\chi_{i, j} = z^B_{S, i, j}/z^B_S, \mu_{i} = z^B_{S, i}/z^B_S$ for every $i \in B, j \in C$. So, for every pair $(\chi, \mu)$ in the support of $\zeta$, we have $\chi_{i,j} \leq \mu_i, \chi_{i, C} \leq u_i \mu_i$ for every $i \in B, j \in C$.  Moreover, either $\mu$ is integral, or $\mu_B \geq \ell_1$.  Since $\sum_{S \in \tcalS}z^B_S  = 1$, $\zeta$ is a distribution over pairs $(\chi, \mu)$. It is not hard to see that $\E_{(\chi, \mu) \sim \zeta}\chi_{i, j} = x_{i, j}$ for every $i \in B, j \in C$ and $\E_{(\chi, \mu) \sim \zeta}\mu_i = y_i$ for every $i \in B$. The support of $\zeta$ has $n^{O(\ell)}$ size. 
	
	\begin{definition}\label{def:good-pairs}
		We say a pair $(\chi, \mu)$ is good if 
		\begin{properties}{def:good-pairs}
			\item $\mu_B \leq y_B/(1-1/\ell)$;  \label{property:good-pairs-facilities}
			\item $\chi_{B, C} \geq (1-1/\ell)x_{B,C}$. \label{property:good-pairs-connections}
		\end{properties}
	\end{definition}
	
	We are only interested in good pairs in the support of $\zeta$.  We show that the total probability of good pairs in the distribution $\zeta$ is large.  Let $\Xi_\rma$ denote the set of pairs $(\chi, \mu)$ satisfying Property~\ref{property:good-pairs-facilities} and $\Xi_\rmb$ denote the set of pairs $(\chi, \mu)$ satisfying Property~\ref{property:good-pairs-connections}.   Notice that $\E_{(\chi, \mu)\sim \zeta}\mu_B = y_B$. By Markov inequality, we have $\sum_{(\chi, \mu) \in \Xi_\rma}\zeta_{\chi, \mu} \geq 1/\ell$. The proof of the following lemma uses elementary mathematical tools.
	\begin{lemma} \label{lemma:zeta-notin-rmb}
		$\sum_{(\chi, \mu) \notin \Xi_\rmb} \zeta_{\chi, \mu} \leq \ell \pi_J/x_{B, C}$. 
	\end{lemma}	

	\begin{proof}
		The idea is to use the property that $J$ is concentrated. To get some intuition, consider the case where $\pi_J = 0$. For every $j \in C$, either $x_{B, j} = 0$ or $x_{B, j} = 1$. Thus, all pairs $(\chi, \mu)$ in the support of $\zeta$ have $\chi_{B, j} = x_{B, j}$ for every $j \in C$; thus $\chi_{B, C} = x_{B, C}$.
		
		Assume towards contradiction that $\sum_{(\chi, \mu) \notin \Xi_\rmb} \zeta_{\chi, \mu} > \ell \pi_J/x_{B, C}$. We sort all pairs $(\chi, \mu)$ in the support of $\zeta$ according to descending order of $\chi_{B, C}$. For any $t \in [0, 1)$, and $j \in C$, define $g_{t, j} \in [0,1]$ as follows. Take the first pair $(\chi, \mu)$ in the ordering such that the total $\zeta$ value of the pairs $(\chi', \mu')$ before $(\chi, \mu)$ in the ordering plus $\zeta_{\chi, \mu}$ is greater than $t$.  Then, define $g_{t, j} = \chi_{B, j}$ and define $g_t = \sum_{j \in C}g_{t, j}  = \chi_{B, C}$.
		
		Fix a client $j \in C$, we have
		\begin{align*}
			x_{B,j} (1-x_{B,j}) = \int_{0}^{x_{B,j}}(1-2t)\sfd t = \int_{0}^{1}\mathbf{1}_{t < x_{B, j}}(1-2t)\sfd t   \geq \int_0^1g_{t, j}(1-2t)\sfd t,
		\end{align*}
		where $\mathbf{1}_{t < x_{B, j}}$ is the indicator variable for the event that $t < x_{B, j}$. The inequality comes from the fact that $\int_0^1 \mathbf{1}_{t < x_{B, j}} \sfd t = x_{B, j} = \int_0^1g_{t, j}\sfd t$, $g_{t, j} \in [0, 1]$ for every $t \in [0, 1)$, and $1-2t$ is a decreasing function of $t$. 
		
		Summing up the inequality over all $j \in C$, we have $\pi_J \geq \int_0^1g_t(1-2t)\sfd t$.  By our assumption that $\sum_{(\chi, \mu) \notin \Xi_\rmb} \zeta_{\chi, \mu} > \ell \pi_J/x_{B, C}$, there exists a number $t^* < 1 - \ell \pi_J/x_{B, C}$ such that $g_t \leq (1-1/\ell)x_{B, C}$ for every $t \in [t^*, 1)$. As $g_t$ is a non-increasing function of $g$ and $\int_0^1 g_t \sfd t = x_{B, C}$, it is not hard to see that $\int_0^1g_t(1-2t)\sfd t$ is minimized when $g_t = (1-1/\ell)x_{B, C}$ for every $t \in [t^*, 1)$ and $g_t = \frac{x_{B,C} - (1-1/\ell)x_{B,C} (1-t^*)}{t^*} = \frac{1/\ell + t^* - t^*/\ell}{t^*}x_{B, C}$ for every $t \in [0,t^*)$.  We have 
		\begin{align*}
			\pi_J &\geq \left(\int_0^{t^*}\frac{1/\ell + t^* - t^*/\ell}{t^*}(1-2t)\sfd t + \int_{t^*}^1(1-1/\ell)(1-2t)\sfd t \right)x_{B, C} \\
			&= \left(\frac{1/\ell + t^* - t^*/\ell}{t^*} \left(t^* - (t^*)^2\right) - (1-1/\ell)\left(t^* - (t^*)^2\right)\right)x_{B, C}\\
			&= \frac{1}{\ell t^*}\left(t^* - (t^*)^2\right)x_{B, C} = \frac{1-t^*}{\ell}x_{B, C} > \frac{\ell\pi_J/x_{B,C}}{\ell}x_{B, C} = \pi_J,
		\end{align*}
		leading to a contradiction.  Thus, we have that $\sum_{(\chi, \mu) \notin \Xi_\rmb} \zeta_{\chi, \mu} \leq \ell\pi_J/x_{B, C}$.   This finishes the proof of Lemma~\ref{lemma:zeta-notin-rmb}.
	\end{proof}
	
	Overall, we have $Q:= \sum_{(\chi, \mu)\text{ good}}\zeta_{\chi, \mu} = \sum_{(\chi, \mu) \in \Xi_\rma \cap \Xi_\rmb}\zeta_{\chi, \mu} \geq 1/\ell - \ell\pi_J/x_{B,C} \geq 1/\ell - 1/(2\ell) = 1/(2\ell)$, 
	where the second inequality used the fact that $\pi_J \leq x_{B,C}/(2\ell^2)$ for $J \in \calJ^\sfC$. 
	
	Now focus on each good pair $(\chi, \mu)$ in the support of $\zeta$. Since $J \in \calJ^\sfC$ and $(\chi, \mu) \in \Xi_\rma$, we have $\mu_B \leq y_B/(1-1/\ell) \leq 2\ell/(1-1/\ell) < \ell_1$ (since we assumed $y_B \leq 2\ell$), if $\ell_1$ is large enough. So, $\mu \in \{0, 1\}^B$.  Then, let $S = \{i\in B: \mu_i = 1 \}$ be the set indicated by $\mu$, and $\beta_i = \chi_{i, C}/(1-1/\ell)$ for every $i \in B$.  For this $(S, \beta)$, Property~\ref{property:concentrated-beta} is satisfied, and we have $\sum_{i \in B}\beta_i = \chi_{B, C}/(1-1/\ell)\geq x_{B,C}$. We then set ${\psi}_{S, \beta}  = \zeta_{\chi, \mu}/Q$.  Thus, ${\psi}$ indeed forms a distribution over pairs $(S, \beta)$.  Moreover, the support of $\zeta$ has size $n^{O(\ell)}$, so does the support of $\psi$. Thus Property~\ref{property:concentrated-support-small} holds. 
	
	Let $s_\psi := \E_{(S, \beta) \sim {\psi}}|S| = \E_{(\chi, \mu)\sim \zeta}\big[\mu_B\big|(\chi, \mu)\text{ good}\big]$. Notice that $\E_{(\chi, \mu) \sim \zeta} \mu_B  = y_B$. By Lemma~\ref{lemma:zeta-notin-rmb}, we have $\sum_{(\chi, \mu) \notin \Xi_\rmb}\zeta_{\chi, \mu} \leq \ell \pi_J/x_{B, C}$.  Thus, $\E_{(\chi, \mu) \sim \zeta}\big[\mu_B\big|(\chi, \mu) \in \Xi_{\rmb}\big] \leq y_B/(1-\ell\pi_j/x_{B, C})$. Since the condition $(\chi, \mu) \in \Xi_\rma$ requires $\mu_B$ to be upper bounded by some threshold, $\E_{(\chi, \mu) \sim \zeta}\big[\mu_B\big|(\chi, \mu) \in \Xi_{\rmb} \cap \Xi_{\rma}\big]$ can only be smaller. Thus, we have that $s_\psi \leq y_B/(1-\ell\pi_J/x_{B,C}) \leq y_B(1+2\ell\pi_J/x_{B,C})$.
	
	The proof of Property~\ref{property:concentrated-cost} for $\psi$ is long and tedious.
	For simplicity, we use $\hat\E[\cdot]$ to denote $\E_{(\chi, \mu) \sim \zeta}\big[\cdot\big|(\chi,\mu)\text{ good}\big]$, and $a = 1/(1-1/\ell)$ to denote the scaling factor we used to define $\beta$. Indeed, we shall lose a factor $O(\ell^2)$ later and thus we shall prove Property~\ref{property:concentrated-cost} for $\psi$ with the $O(\ell^2)$ term on the right:
	\begin{lemma}\label{lemma:concentrated-cost}
		$\hat\E[\EMD_J(\alpha|_J, \beta)] \leq O(\ell^2)D_B$, where $\beta$ depends on $\chi$ as follows: $\beta_i = a\chi_{i, C}$ for every $i \in B$. 
	\end{lemma}

	\begin{proof}
		Focus on a good pair $(\chi, \mu)$ and the $\beta$ it defined: $\beta_i = a\chi_{i, C}$ for every $i \in B$. We call $\alpha$ the demand vector and $\beta$ the supply vector.   Since $(\chi, \mu)$ is good, $\sum_{i \in B} \beta_i = a\chi_{B, C} \geq x_{B, C}  = \sum_{v \in J} \alpha_v$. Thus we can satisfy all the demands and $\EMD(\alpha, \beta)$ is not $\infty$. 
		
		We satisfy the demands in two steps.  In the first step, we give colors to the supplies and demands; each color is correspondent to a client $j \in C$.  Notice that $\alpha_v = \sum_{j \in C}x_{U_v,j}$ and $\beta_i = a\sum_{j \in C}\chi_{i, j}$.  For every $v \in J, j \in C$, $x_{U_v, j}$ units of demand at $v$ has color $j$; for every $i \in B, j \in C$, $a\chi_{i,j}$ units of supply at $i$ have color $j$. In this step, we match the supply and demand using the following greedy rule: while for some $j \in C, i, i' \in B$, there is unmatched demand of color $j$ at $v$ and there is unmatched supply of color $j$ at $i$, we match them as much as possible.  The cost for this step is at most the total cost of moving all supplies and demands of color $j$ to $j$, i.e, 
		\begin{align*}
			&\quad \sum_{v \in J, i \in U_v, j \in C}x_{i, j}(d(v, i)  + d(i, j)) + a\sum_{i \in B, j \in C}\chi_{i,j}d(i, j) \\
			&\leq \sum_{v \in J, i \in U_v}x_{i, C}d(v, i) + \sum_{i \in B, j \in C}(x_{i, j} + a\chi_{i, j})d(i, j) \\
			&\leq O(1)D_B+ \sum_{i \in B, j \in C}(x_{i, j} + a\chi_{i, j})d(i, j), \qquad \qquad \text{by Lemma~\ref{lemma:moving-to-representatives}. }
		\end{align*}
		
		After this step, we have $\sum_{j \in C}\max\{x_{B, j} - a\chi_{B, j}, 0\} \leq \sum_{j \in C}\max\{x_{B, j} - \chi_{B, j}, 0\}$ units of unmatched demand. 
		
		In the second step,  we match remaining demand and the supply.  For every $v \in J, i \in U_v$, we move the remaining supply at $i$ to $v$.  After this step, all the supplies and the demands are at $J$; then we match them arbitrarily.  The total cost is at most 
		\begin{align}
			\sum_{v \in J, i \in U_v}a\chi_{i, C} d(i, v) + \sum_{j \in C}\max\{x_{B, j} - \chi_{B, j}, 0\} \times \diam(J), \label{quan:cost-second-step}
		\end{align}
		where $\diam(J)$ is the diameter of $J$. 
		
		Notice that $\hat\E[\chi_{i, j}] \leq 2\ell x_{i, j}$ since $\Pr_{(\chi, \mu)\sim \zeta}[(\chi, \mu)\text{ good}] \geq 1/(2\ell$)  and $\E_{(\chi, \mu)\sim \zeta}\chi_{i, j} = x_{i, j}$.  The expected cost of the first step is at most $O(1)D_B + O(\ell)\sum_{j \in C, i \in B}x_{i, j}d(i,j) = O(\ell)D_B$.  Similarly, the expected value of the first term of \eqref{quan:cost-second-step} is at most $O(\ell)\sum_{v \in J, i \in U_v}x_{i,C}d(i,v) \leq O(\ell)D_B$ by Lemma~\ref{lemma:moving-to-representatives}.
		
		Consider the second term of \eqref{quan:cost-second-step}. Notice that $\hat\E[\max\{x_{B, j} - \chi_{B, j}, 0\}] \leq x_{B, j}$.  Also, 
		\begin{align*}
			\hat\E[\max\{x_{B, j} - \chi_{B, j}, 0\}] &=  \hat\E[\max\{ (1 - \chi_{B, j}) - (1 - x_{B, j}), 0\}]  \\
			&\leq \hat\E[1 - \chi_{B, j}] \leq 2\ell(1-x_{B, j}).
		\end{align*}
		
		So, $\hat\E\max\{x_{B, j} - \chi_{B, j}, 0\} \leq \min\{x_{B, j}, 2\ell(1-x_{B, j})\} \leq 3\ell x_{B,j}(1-x_{B,j})$: if $x_{B, j} \geq 1-1/(2\ell) \geq 2/3$, then we have $2\ell(1-x_{B, j}) \leq 2\ell(1-x_{B, j})\cdot(3x_{B, j}/2) = 3\ell x_{B, j}(1-x_{B,j})$; if $x_{B, j} < 1-1/(2\ell)$, then $1 - x_{B, j} > 1/(2\ell)$, implying $x_{B, j} \leq 2\ell x_{B, j}(1-x_{B, j})$.
		
		Summing up the inequality over all clients $j \in C$, we have $\hat\E\big[\sum_{j \in C}\max\{x_{B, C}-\chi_{B,C},0 \}\big] \leq O(\ell)\pi_J$. So, the expected value of the second term of \eqref{quan:cost-second-step} is at most $O(\ell)\pi_J \cdot \diam(J) \leq O(\ell^2) \pi_J L(J) \leq O(\ell^2)D_B$, by Lemma~\ref{lemma:x-j-times-one-minus-x-j-small}.  This finishes the proof of Lemma~\ref{lemma:concentrated-cost}. 
	\end{proof}
	
	At this point, we may have $s_\psi < y_B$.  We can apply the following operation repeatedly. Take a pair $(S, \beta)$ with $\psi_{S, \beta} > 0$ and $S \subsetneq B$.  We then shift some $\psi$-mass from the pair $(S, \beta)$ to $(B, \beta)$ so as to increase $s_\psi$. Thus, we can assume Property~\ref{property:concentrated-S-expectation} holds for $\psi$. 
	
	Property~\ref{property:concentrated-enough-demand} may be unsatisfied: we only have $\sum_{i \in B}\beta_i \geq x_{B, C}$ for every $(S, \beta)$ in the support of $\psi$. To satisfy the property, we focus on each $(S, \beta)$ in the support of $\psi$ such that $\sum_{i \in B}\beta_i > x_{B, C}$. By considering the matching that achieves $\EMD(\alpha|_J, \beta)$, we can find a $\beta' \in \R_{\geq 0}^B$ such that $\beta'_i \leq \beta_i$ for every $i \in B$, $\sum_{i \in B} \beta'_i = \sum_{v \in J}\alpha_v = x_{B,C}$, and $\EMD(\alpha|J, \beta') = \EMD(\alpha|J, \beta)$.  We then shift all the $\psi$-mass at $(S, \beta)$ to $(S, \beta')$.
	
	To sum up what we have so far, we have a distribution $\psi$ over $(S, \beta)$ pairs, that satisfies Properties~\ref{property:concentrated-S-expectation}, \ref{property:concentrated-beta}, \ref{property:concentrated-enough-demand}, \ref{property:concentrated-support-small} and Property~\ref{property:concentrated-cost} with $O(\ell^4)$ replaced with $O(\ell^2)$.  The only Property that is missing is Property~\ref{property:concentrated-S}; to satisfy the property, we shall apply the following lemma to massage the distribution $\psi$.
	

	
	\begin{lemma}\label{lemma:modify-distribution}
		Given a distribution $\psi$ over pairs $(S \subseteq B, \beta \in [0, 1]^B)$ satisfying $s_\psi := \E_{(S, \beta)\sim \psi}|S| \leq \ell_1$, 
		we can construct another distribution $\psi'$ such that 
		\begin{properties}{lemma:modify-distribution}
			\item $\psi'_{S, \beta} \leq O(\ell^2)\psi_{S, \beta}$ for every pair $(S, \beta)$; \label{property:massage-keep-marginal}
			\item every pair $(S, \beta)$ in the support of $\psi'$ has $|S| \leq \ceil{s_\psi}$; \label{property:massage-upper-threshold}
			\item $\E_{(S, \beta) \sim \psi'} \max\{|S|, \floor{s_\psi}\} \leq s_\psi$. \label{property:massage-lower-threshold}
		\end{properties}
	\end{lemma}
	
	Property~\ref{property:massage-keep-marginal} requires that the probability that a pair $(S, \beta)$ happens in $\psi'$ can not be too large compared to the probability it happens in $\psi$. 
	Property~\ref{property:massage-upper-threshold} requires $|S| \leq \ceil{s_\psi}$ for every $(S, \beta)$ in the support of $\psi'$. Property~\ref{property:massage-lower-threshold} corresponds to requiring $|S| \geq \floor{s_\psi}$: even if we count the size of $S$ as $\floor{s_\psi}$ if $|S| \leq \floor{s_\psi}$,  the expected size is still going to be at most $s_\psi$.  
	
	\begin{proof}[Proof of Lemma~\ref{lemma:modify-distribution}]
		If $s_\psi - \floor{s_\psi} \leq 1-1/\ell$, then we shall throw away the pairs with $|S| > \floor{s_\psi}$.   More formally, let $Q = \Pr_{(S, \beta)\sim \psi} \big[|S| \leq \floor{s_\psi}\big]$ and we define $\psi'_{S, \beta} = \psi_{S, \beta}/Q$ if $|S| \leq \floor{s_\psi}$ and $\psi'_{S, \beta} = 0$ if $|S| \geq \floor{s_\psi} + 1$.  So, Property~\ref{property:massage-upper-threshold} is satisfied. By Markov inequality, we have that $Q \geq 1- s_\psi/(\floor{s_\psi} + 1) = (\floor{s_\psi} - s_\psi + 1)/(\floor{s_\psi} + 1)\geq (1/\ell)/(\floor{s_\psi}+1) \geq 1/(\ell\ell_1+\ell)$ since $s_\psi \leq \ell_1 = O(\ell)$.  Thus, $\psi'_{S, \beta} \leq O(\ell^2)\psi_{S, \beta}$ for every pair $(S, \beta)$, implying Property~\ref{property:massage-keep-marginal}. 
		Every pair $(S, \beta)$ in the support of $\psi'$ has $|S| \leq \floor{s_\psi}$ and thus Property~\ref{property:massage-lower-threshold} holds.
		
		Now, consider the case where $s_\psi - \floor{s_\psi} > 1-1/\ell$. In this case, $s_\psi$ is a fractional number. Let $\psi''$ be the distribution obtained from $\psi$ by conditioning on pairs $(S, \beta)$ with $|S| \leq \ceil{s_\psi}$. By Markov inequality, we have $\Pr_{(S, \beta)\sim \psi} \big[|S| \leq \ceil{s_\psi}\big] \geq 1-s_\psi/(\ceil{s_\psi} + 1) \geq 1-s_\psi/(s_\psi+1) \geq 1/(\ell_1+1)$ as $s_\psi \leq \ell_1 = O(\ell)$. So,  $\psi''_{S, \beta} \leq O(\ell)\psi_{S, \beta}$ for every pair $(S, \beta)$. 
		Moreover, we have $\E_{(S, \beta) \sim \psi''}|S| \leq s_\psi$ since we conditioned on the event that $|S|$ is upper-bounded by some-threshold; all pairs $(S, \beta)$ in the support of $\psi''$ have $|S| \leq \ceil{s_\psi}$.
		
		Then we modify $\psi''$ to obtain the final distribution $\psi'$. Notice that for a pair $(S, \beta)$ with $|S| \leq \floor{s_\psi}$, we have $s_\psi - |S| \leq s_\psi \leq 2\ell_1(s_\psi-\floor{s_\psi})$. Thus,
		\begin{align*}
			\sum_{(S, \beta): |S| \leq \floor{s_\psi}} \psi''_{S, \beta}(s_\psi - \floor{s_\psi}) &\geq \frac{1}{2\ell_1} \sum_{(S, \beta): |S| \leq \floor{s_\psi}} \psi''_{S, \beta}(s_\psi - |S|)\\
			& \geq \frac{1}{2\ell_1}\sum_{(S, \beta):|S| = \ceil{s_\psi}}\psi''_{S, \beta}(\ceil{s_\psi}-s_\psi),
		\end{align*}
		where the second inequality is due to $\E_{(S, \beta)\sim \psi''}|S| \leq s_\psi$.
		
		For every pair $(S, \beta)$ with $|S| \leq \floor{s_\psi}$, let $\psi'_{S, \beta}  = \psi''_{S, \beta}$. 
		For every pair $(S, \beta)$ such that $|S| = \ceil{s_\psi}$, we define $\psi'_{S, \beta} = \psi''_{S, \beta}/(2\ell_1)$.   Due to the above inequality, we have $\sum_{(S, \beta): |S| \leq \floor{s_\psi}}\psi'_{S, \beta}(s_\psi - \floor{s_\psi}) \geq \sum_{(S, \beta):|S| = \ceil{s_\psi}}\psi'_{S, \beta}(|S| - s_\psi)$, implying $\sum_{(S, \beta)}\psi'_{S, \beta}\max\{|S| -s_\psi, \allowbreak \floor{s_\psi} - s_\psi\} \leq 0$. Finally, we scale the $\psi'$ vector so that we have $\sum_{(S, \beta)}\psi'_{S, \beta} = 1$; Properties~\ref{property:massage-upper-threshold} and \ref{property:massage-lower-threshold} hold. The scaling factor is at most $2\ell_1 = O(\ell)$. Overall, we have $\psi'_{S, \beta} \leq O(\ell^2)\psi_{S, \beta}$ for every pair $(S, \beta)$ and Property~\ref{property:massage-keep-marginal} holds. This finishes the proof of Lemma~\ref{lemma:modify-distribution}.
	\end{proof}	
	
	With Lemma~\ref{lemma:modify-distribution} we can finish the proof of Theorem~\ref{thm:concentrated-sets} for the case $y_B \leq 2\ell$.  We apply the lemma to $\psi$ to obtain the distribution $\psi'$. By Property~\ref{property:massage-keep-marginal}, Properties~\ref{property:concentrated-beta}, \ref{property:concentrated-enough-demand} and \ref{property:concentrated-support-small} remain satisfied for $\psi'$;  Property~\ref{property:concentrated-cost} also holds for $\psi'$, as we lost a factor of $O(\ell^2)$ on the expected cost. 
	
	To obtain our final distribution $\phi$, initially we let $\phi_{S, \beta} = 0$ for every pair $(S, \beta)$.  For every $(S, \beta)$ in the support of $\psi'$, we apply the following procedure.  If $|S| \geq \floor{s_\psi}$, then we increase $\phi_{S, \beta}$ by $\psi'_{S, \beta}$; otherwise, take an arbitrary set $S' \subseteq B$ such that $S \subseteq S'$ and $|S'| = \floor{s_\psi}$ and increase $\phi_{S', \beta}$ by $\psi_{S, \beta}$. Due to Property~\ref{property:massage-upper-threshold}, every pair $(S, \beta)$ in the support of $\phi$ has $|S| \in \{\floor{s_\psi}, \ceil{s_\psi}\}$. 
	Property~\ref{property:massage-lower-threshold} implies that $s_\phi:=\E_{(S, \beta)\sim \phi}|S| \leq s_\psi \in [y_B, (1+2\ell\pi)y_B]$.  If $s_\phi < s_\psi$, we increase $s_\phi$ using the following operation. Take an arbitrary pair $(S, \beta)$ in the support of $\phi$ such that $|S| = \floor{s_\psi}$, let $S' \supseteq S$ be a set such that $S' \subseteq B$ and $|S'| = \ceil{s_\psi}$, we decrease $\phi_{S, \beta}$ and increase $\phi_{S', \beta}$. Eventually, we can guarantee $s_\phi = s_\psi$; thus Properties~\ref{property:concentrated-S-expectation} and \ref{property:concentrated-S} are satisfied. This finishes the proof of Theorem~\ref{thm:concentrated-sets} when $y_B \leq 2\ell$. 

	Now we handle the case where $y_B > 2\ell$. By Properties~\ref{property:black-weights-of-components} and \ref{property:black-root-is-big}, $J$ is a root black component that contains a single representative $v$ and $y_{U_v=B} > 2\ell$. First we find a nearly integral solution with at most $2\ell +2$ open facilities. Then we close two facilities serving the minimum amount of demand and spread their demand among the remaining facilities. Since there is at least $2\ell$ open facilities remaining, we increase the amount of demand at any open facility by no more than a factor of $O(1/\ell)$.   
	
	Let $u'_i = \frac{x_{i,C}}{y_i } \leq u_i$. We may scale $u'_i$ by a factor of $1 + O(1/\ell)$ during the course of the algorithm. Consider the following LP with variables $\{\lambda_i\}_{i \in B}$:
	\begin{align}
		\text{min}\quad \sum_{i \in U_v}u'_i\lambda_i d(i, v) \quad \text{s.t.} 				\label{LP:concentrated-large-y}
	\end{align}
	\vspace*{-0.3\abovedisplayskip}			
	\begin{align*}
		\sum_{i \in B}u'_i\lambda_i = x_{B, C}; \quad 
		\sum_{i \in B}\lambda_i &= y_B; \quad
		\lambda_i \in [0,1],\quad \forall i  \in B. 
	\end{align*}
	
	By setting $\lambda_i = y_i$, we obtain a solution to LP\eqref{LP:concentrated-large-y} of value $\sum_{i \in U_v}x_{i,C}d(i,v) \leq O(1)D_{U_v}$, by Lemma~\ref{lemma:moving-to-representatives}. So, the value of LP\eqref{LP:concentrated-large-y} is at most $O(1)D_{U_v}$. Fix on such an optimum vertex-point solution $\lambda$ of LP\eqref{LP:concentrated-large-y}. Since there are only two non-box-constraints, $\lambda$ has at most two fractional $\lambda_i$.  Moreover, as $y_{U_v} \geq 2\ell$,  there are at least $2\ell$ facilities in the support of $\lambda$. 
	
	We shall reduce the size of the support of $\lambda$ by $2$, by repeating the following procedure twice. Consider the $i^*$ in the support with the smallest $\lambda_{i^*} u'_{i^*}$ value.  Let $a := \lambda_{i^*} u'_{i^*}/ \sum_{i \in U_v}\lambda_{i}u'_{i} \leq O(\frac{1}{\ell})$, we then scale $u'_{i}$ by a factor of $1/(1-a) \leq 1 + O(1/\ell)$ for every $i \in U_v \setminus i^*$ and change $\lambda_{i^*}$ to $0$.  So, we still have $\sum_{i \in U_v}u'_i \lambda_i = x_{U_v, C}$. The value of the objective function is scaled by a factor of at most $1+O(1/\ell)$.
	
	Let $S = \{i \in U_v: \lambda_i > 0\}$ and let $\beta_i = \lambda_iu'_i$ for every $i \in U_v$.  So, $|S| \leq y_{U_v}$ Properties~\ref{property:concentrated-beta} and \ref{property:concentrated-enough-demand} are satisfied.  Moreover, $\EMD_J(\alpha|_J, \beta) \leq \sum_{i \in U_v}\beta_id(i, v) \leq O(1)D_{U_v}$ since the value of LP\eqref{LP:concentrated-large-y} is $O(1)D_{U_v}$ and we have scaled each $\beta_i$ by at most a factor of $1+O(1/\ell)$. 
	
	If we let $\phi$ contains the single pair $(S, \beta)$ with probability $1$, then all properties from \ref{property:concentrated-beta} to \ref{property:concentrated-cost} are satisfied.  To satisfy Properties~\ref{property:concentrated-S-expectation}  and \ref{property:concentrated-S}, we can manually add facilities to $S$ with some probability, as we did before for the case $y_B \leq 2\ell$.   This finishes the proof of Theorem~\ref{thm:concentrated-sets}.		
\fi
\ifdefined\proceeding
\else
	\subsection{Local Solutions for Unions of Non-Concentrated Components}
	\label{subsec:distri-non-con}
		In this section, we construct a local solution for the union $V$ of some close non-concentrated black components.
\fi
		\begin{lemma}\label{lemma:non-concentrated-sets}
			Let $\calJ' \subseteq \calJ^\sfN$ be a set of non-concentrated black components, $V = \bigcup_{J \in \calJ'} J$ and $B = U(V)$.  Assume there exists $v^* \in R$ such that $d(v, v^*) \leq O(\ell^2)L(J)$ for every $J \in \calJ'$ and $v \in J$. Then, we can find a pair $(S \subseteq B, \beta \subseteq \R_{\geq 0}^B)$ such that 
			\begin{properties}{lemma:non-concentrated-sets}
				\item $|S|  \in \big\{\ceil{y_B}, \ceil{y_B} + 1\big\}$; \label{property:non-concentrated-S}
				\item $\beta_i \leq u_i $ if $i \in S$ and $\beta_i = 0$ if $i \in B \setminus S$; \label{property:non-concentrated-beta}
				\item $\sum_{i \in S} \beta_i = x_{B, C} = \sum_{v \in V}\alpha_v$; \label{property:non-concentrated-enough-demand}
				\item $\EMD_V(\alpha|_V, \beta) \leq O(\ell^2\ell_2)D_B$. \label{property:non-concentrated-cost}
			\end{properties}
		\end{lemma}

\ifdefined\proceeding
\else
	\begin{proof}
		We shall use an algorithm similar to the one we used for handling the case where $y_B > 2\ell$ in Section~\ref{subsec:distri-con}. Again, for simplicity, we let $u'_i = \frac{x_{i,C}}{y_i } \leq u_i$ to be the ``effective capacity'' of $i$.  Consider the following LP with variables $\{\lambda_i\}_{i \in B}$:
		\begin{align}
			\text{min}\quad \sum_{J \in \calJ', v \in J, i \in U_v}u'_i\lambda_i \big(d(i, v) +\ell^2  L(J)\big) \quad \text{s.t.} 				\label{LP:non-concentrated}
		\end{align}
		\vspace*{-1\abovedisplayskip}			
		\begin{align*}
			\sum_{i \in B}u'_i\lambda_i = x_{B, C}; \quad 
			\sum_{i \in B}\lambda_i &= y_B; \quad
			\lambda_i \in [0,1],\quad \forall i  \in B. 
		\end{align*}
		
		By setting $\lambda_i = y_i$, we obtain a valid solution to the LP with the objective value
		\begin{align}
			&\quad \sum_{J \in \calJ', v \in J, i \in U_v}x_{i,C}\big(d(i,v) + \ell^2 L(J)\big) \nonumber\\
			&\leq \sum_{v \in V, i \in U_v}x_{i,C}d(i,v) + \ell^2 \sum_{J \in \calJ'}x_{U(J), C}L(J) 
			\leq \sum_{v \in V}O(1)D_{U_v} + \ell^2 \sum_{J \in \calJ'}\ell_2\pi_J L(J) \nonumber\\
			&\leq O(1)D_B + \ell^2 \ell_2\sum_{J \in \calJ'} O(1)D_{U(J)}= O(\ell^2 \ell_2)D_B, \label{equ:LP-value-non-c}
		\end{align} 
		by Lemma~\ref{lemma:moving-to-representatives} and Lemma~\ref{lemma:x-j-times-one-minus-x-j-small}. So, the value of LP\eqref{LP:non-concentrated} is at most $O(\ell^2\ell_2)D_B$.
		
		Fix such an optimum vertex-point solution $\lambda$ of LP \eqref{LP:non-concentrated}. Since there are only two non-box-constraints, every vertex-point $\lambda$ of the polytope has at most two fractional $\lambda_i$. 
		
		Let $S = \{i \in B: \lambda_i > 0\}$ and let $\beta_i = \lambda_iu'_i$ for every $i \in B$.  So, Properties~\ref{property:non-concentrated-S}, \ref{property:non-concentrated-beta} and \ref{property:non-concentrated-enough-demand} are satisfied. 
		
		Now we prove Property~\ref{property:non-concentrated-cost}. To compute $\EMD_V(\alpha|_V, \beta)$, we move all demands in $\alpha|_V$ and all supplies in $\beta$ to $v^*$. The cost is
		\begin{align*}
			&\quad \sum_{v \in V}\alpha_v d(v, v^*) + \sum_{i \in B}\beta_i d(i, v^*) \\
			&\leq \sum_{J \in \calJ', v \in J} x_{U_v, C} O(\ell^2) L(J) + \sum_{J \in \calJ', v \in J, i \in U_v}\beta_i \big(d(i,v) + O(\ell^2)L(J)\big)\leq O(\ell^2\ell_2)D_B,
		\end{align*}
		where the $O(\ell^2\ell_2) D_B$ for the first term was proved in \eqref{equ:LP-value-non-c} and the bound $O(\ell^2\ell_2)D_B$ for the second term is due to the fact that $\gamma$ is an optimum solution to LP\eqref{LP:non-concentrated}.
		This finishes the proof of Lemma~\ref{lemma:non-concentrated-sets}.
	\end{proof}
\fi

%% file: rounding.tex
\section{Rounding Algorithm}
\label{sec:rounding}

	In this section we describe our rounding algorithm. We start by giving the intuition behind the algorithm. For each concentrated component $J \in \calJ$, we construct a distribution of local solutions using Theorem~\ref{thm:concentrated-sets}. We shall construct a partition $\calV^\sfN$ of the representatives in $\union_{J \in \calJ^\sfN}J$ so that each $V \in \calV^\sfN$ is the union of some nearby components in $\calJ^\sfN$.  For each set $V \in \calV^\sfN$, we apply Lemma~\ref{lemma:non-concentrated-sets} to construct a local solution.  If we independently and randomly choose a local solution from every distribution we constructed, then we can move all the demands to the open facilities at a small cost, by Property~\ref{property:concentrated-cost} and Property~\ref{property:non-concentrated-cost}. 
	
	However, we may open more than $k$ facilities, even in expectation. Noticing that the fractional solution opens $y_B$ facilities in a set $B$, the extra number of facilities come from two places.  In Property~\ref{property:concentrated-S-expectation} of Theorem~\ref{thm:concentrated-sets}, we may open in expectation $y_B\cdot 2\ell\pi_J/x_{B, C}$ more facilities in $B$ than $y_B$. Then in Property~\ref{property:non-concentrated-S} of Lemma~\ref{lemma:non-concentrated-sets}, we may open $\ceil{y_B}$ or $\ceil{y_B}+1$ facilities in $B$. To reduce the number of open facilities to $k$, we shall shut down (or remove) some already-open facilities and move the demands satisfied by these facilities to the survived open facilities: a concentrated component $J \in \calJ^\sfC$ is responsible for removing $y_B\cdot 2\ell\pi_J/x_{B, C} < 1$ facilities in expectation; a set $V \in \calV^\sfN$ is responsible for removing up to $2$ facilities.  Lemma~\ref{lemma:x-j-times-one-minus-x-j-small} allows us to bound the cost of moving demands caused by the removal, provided that the moving distance is not too big.  To respect the capacity constraint up to a factor of $1+\eps$,  we are only allowed to scale the supplies of the survived open facilities by a factor of $1+O(1/\ell)$.  Both requirements will be satisfied by the forest structure over groups and the fact that each non-leaf group contains $\Omega(\ell)$ fractional opening (Property~\ref{property:groups-non-leaf-big}).  Due to the forest structure and Property~\ref{property:groups-non-leaf-big}, we always have enough open facilities locally that can support the removing of facilities.
	
	In order to guarantee that we always open $k$ facilities, we need to use a dependent rounding procedure for opening and removing facilities. As in many of previous algorithms, we incorporate the randomized rounding procedure into random selections of vertex points of polytopes respecting marginal probabilities.  In many cases, a randomized selection procedure can be derandomized since there is an explicit linear objective we shall optimize. 
	
	\begin{figure}
		\centering
		\includegraphics[width=0.8\textwidth]{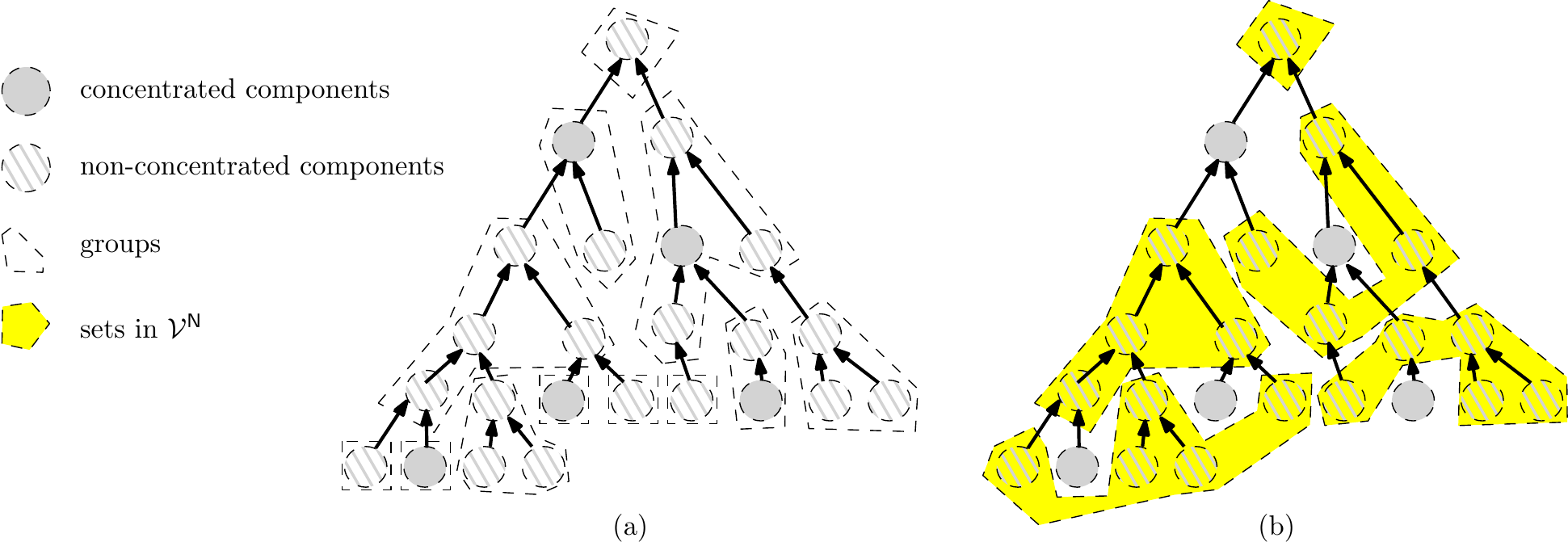}
		\caption{Figure (a) gives the forest $\Upsilon^*_{\calJ}$ over $\calJ$ and the set $\calG$ of groups (denoted by empty polygons). Figure (b) gives $\calV^\sfN$: each set $V \in \calV^\sfN$ is the union of components in a solid polygon.}
		\label{fig:VN}
	\end{figure}
		
	We now formally describe our rounding algorithm.  For every group $G \in \calG$, we use $\Lambda_G$ to denote the set of child-groups of $G$.  We construct a partition $\bbJ^\sfC$ of $\calJ^\sfC$ as follows.  For each root group $G \in \calG$,  we add $G \cap \calJ^\sfC$ to $\bbJ^\sfC$ if it is not empty.  For each non-leaf group $G \in \calG$, we add $\union_{G' \in \Lambda_G} (G' \cap \calJ^\sfC)$ to $\bbJ^\sfC$, if it is not empty.   We construct the partition $\bbJ^\sfN$ for $\calJ^\sfN$ in the same way, except that we consider components in $\calJ^\sfN$.  We also define a set $\calV^\sfN$ as follows: for every $\calJ' \in \bbJ^\sfN$, we add $\union_{J \in \calJ'} J$ to $\calV^\sfN$; thus, $\calV^\sfN$ forms a partition for $\union_{J \in \calJ^\sfN}J$. See Figure~\ref{fig:VN} for the definition of $\calV^\sfN$. 
	
	In Section~\ref{subsec:initial-open-facilities}, we describe the procedure for opening a set $S^*$ of facilities, whose cardinality may be larger than $k$.   Then in Section~\ref{subsec:remove}, we define the procedure $\remove$, which removes one open facility. We wrap up the algorithm in Section~\ref{subsec:final-solution}.
		
	\subsection{Constructing Initial Set $S^*$ of Open Facilities} 
	\label{subsec:initial-open-facilities}
	
	In this section, we open a set $S^*$ of facilities, whose cardinality may be larger than $k$,  and construct a supply vector $\beta^* \in \R_{\geq 0}^F$ such that $\beta^*_i = 0$ if $i \notin S^*$. $(S^*, \beta^*)$ will be the concatenation of all local solutions we constructed.
	
	It is easy to construct local solutions for non-concentrated components.  For each set $\calJ' \in \bbJ^\sfN$ of components and its correspondent $V = \union_{J \in \calJ'}J \in \calV^\sfN$, we apply Lemma~\ref{lemma:non-concentrated-sets} to obtain a local solution $\big(S \subseteq U(V), \beta \in \R_{\geq 0}^{U(V)}\big)$. Then, we add $S$ to $S^*$ and let $\beta^*_i = \beta_i$ for every $i \in U(V)$.  Notice that $\calJ'$ either contains a single root black component $J$, or contains all the non-concentrated black components in the child-groups of some group $G$. In the former case, the diameter of $J$ is at most $O(\ell)L(J)$ by Property~\ref{property:black-edges-short}; in the latter case, we let $v^*$ be an arbitrary representative in $\union_{J' \in G}J'$ and then any representative $v \in J, J \in \calJ'$ has $d(v, v^*) \leq O(\ell^2)L(J)$ by Property~\ref{property:groups-close}.  Thus, all the properties in Lemma~\ref{lemma:non-concentrated-sets} are satisfied.
	
	For concentrated components, we only obtain distributions of local solutions by applying Theorem~\ref{thm:concentrated-sets}. For every $J \in \calJ^\sfC$, we check if Constraints~\eqref{CLPC:add-to-one} to \eqref{CLPC:more-than-ell-facilities} are satisfied for $B = U(J)$. If not, we return a separation plane for the fractional solution; otherwise we apply Theorem~\ref{thm:concentrated-sets} to each component $J$ to obtain a distribution $\big(\phi^J_{S, \beta}\big)_{S \subseteq U(J), \beta \in \R_{\geq 0}^{U(J)}}$. To produce local solutions for concentrated components, we shall use a dependent rounding procedure that respects the marginal probabilities. As mentioned earlier, we shall define a polytope and the procedure randomly selects a vertex point of the polytope. 
	
	We let $s_J := s_{\phi_J} := \E_{(S, \beta) \sim \phi^J}|S|$ be the expectation of $|S|$ according to distribution $\phi^J$.  For notational convenience, we shall use $a \approx b$ to denote $a \in \big[\floor{b}, \ceil{b}\big]$.  Consider the following polytope $\calP$ defined by variables $\{\psi^J_{S, \beta}\}_{J \in \calJ^\sfC, S, \beta}$ and $\{q_J\}_{J \in \calJ^\sfC}$.\footnote{For every $J \in \calJ^\sfC$, we only consider the pairs $(S, \beta)$ in the support of $\phi^J$; thus the total number of variables is $n^{O(\ell)}$.} 
	
\noindent \hspace*{-0.04\textwidth} \begin{minipage}{0.46\textwidth}
	 	\begin{alignat}{2}
	 		\psi^J_{S, \beta}, p_J &\in [0, 1] &\  \forall J &\in \calJ^\sfC, S, \beta; \label{LPC:0-or-1}\\[1.1em]
	 		\sum_{S, \beta}\psi^J_{S,\beta} &= 1, &\  \forall J &\in \calJ^\sfC; \label{LPC:one-config-1}
	 	\end{alignat}
	 \end{minipage}
	 \begin{minipage}{0.53\textwidth}
	 	\begin{alignat}{2}
	 		\sum_{J \in \calJ'} q_J &\leq 1, &\  \forall \calJ' &\in \bbJ^\sfC; \label{LPC:q-partition} \\
	 		\sum_{S, \beta} \psi^J_{S, \beta} |S|  - q_J &\approx y_{U(J)}, &\  \forall J &\in \calJ^\sfC; \label{LPC:each-component}
	 	\end{alignat}
	 \end{minipage}
	
	 \vspace*{-8pt}
	 
	 \begin{alignat}{2}
	 		\sum_{J \in \calJ'}\Big(\sum_{S, \beta} \psi^J_{S, \beta} |S|  - q_J\Big) &\approx \sum_{J \in \calJ'} y_{U(J)}, &\quad \forall \calJ' &\in \bbJ^\sfC; \label{LPC:partition-components} \\
	 		\sum_{J \in \calJ^\sfC}\Big(\sum_{S, \beta} \psi^J_{S, \beta} |S| - q_J\Big) &\approx \sum_{J \in \calJ^\sfC} y_{U(J)}. \label{LPC:all-components}
	 \end{alignat}
	 		 
	In the above LP, $\psi^J$ is the indicator vector for local solutions for $J$ and $q_J$ indicates whether $J$ is responsible for removing one facility; if $q_J = 1$, we shall call $\remove(J)$ later.  Up to changing of variables, any vertex point of $\calP$ is defined by two laminar families of tight constraints and thus $\calP$ is integral: 
	\begin{lemma}\label{lemma:P-integral}
		$\calP$ is integral.
	\end{lemma}
	
\ifdefined\proceeding
\else
	\begin{proof}
		To avoid negative coefficients, we shall let $q'_J = 1 - q_J$ and focus on $\psi$ and $q'$ variables.  Consider the set of tight constraints that define a vertex point.  The tight constraints from \eqref{LPC:0-or-1}, \eqref{LPC:one-config-1} and \eqref{LPC:q-partition} define a matroid base polytope for a laminar matroid.
		
		For each $J \in \calJ$, every pair $(S, \beta)$ in the support of $\phi^J$ has $|S| \approx s_J$.  Thus, Constraint~\eqref{LPC:each-component} is equivalent to $\sum_{S, \beta}\psi^J_{S, \beta}(|S| - \floor{s_J}) + q'_J \approx y_{U(J)}  + 1 - \floor{s_J}$. This is true since Constraint~\eqref{LPC:one-config-1} holds and $\psi^J$ is a distribution.  We do the same transformation for Constraints~\eqref{LPC:partition-components} and \eqref{LPC:all-components}. 
		It is easy to see that the tight constraints from \eqref{LPC:each-component}, \eqref{LPC:partition-components} and \eqref{LPC:all-components} also define the matroid-base-polytope for a laminar matroid.
		
		Thus, by the classic matroid theory, the set of tight constraints define an integral solution; thus $\calP$ is integral. 
	\end{proof}
\fi
	
	We set $\psi^{*J}_{S, \beta} = \phi^J_{S, \beta}$ and $q^*_J = s_J - y_{U(J)}$ for every $J \in \calJ^\sfC$ and $(S, \beta)$. Then, 
	\begin{lemma}
		$(\psi^*, q^*)$ is a point in polytope $\calP$. 
	\end{lemma}
\ifdefined\proceeding
\else
	\begin{proof}
		Notice that $s_J$ is the expected size of $S$ according to the distribution $\phi^J$, while $y_{U(J)}$ is the budget for the number of open facilities open in $U(J)$.  So $q^*_J$ is the expected number of facilities that go beyond the budget. It is easy to see that Constraints~\eqref{LPC:0-or-1} and \eqref{LPC:one-config-1} hold for $\psi^*$ and $q^*$. For every $\calJ' \in \bbJ^\sfC$, we have that $\sum_{J \in \calJ'}q^*_J \leq 2\ell\sum_{J \in \calJ'}y_{U(J)}\pi_J/x_{U(J), C} \leq 2\ell\sum_{J \in \calJ'}y_{U(J)}/\ell_2 \leq 2\ell \times O(\ell^2)/\ell_2$ due to Properties~\ref{property:groups-few-children} and \ref{property:concentrated-S-expectation}. This at most 1 if $\ell_2 = \Theta(\ell^3)$ is large enough. (If $\calJ'$ contains a root component $J$ which has $y_{U(J)} > 2\ell$ then $\sum_{J \in \calJ'} q^*_J = 0$.) Thus, Constraint~\eqref{LPC:q-partition} holds. $\sum_{S, \beta}\psi^{*J}_{S, \beta}|S| - q^*_J = s_J - q^*_J = y_{U(J)}$. So, Constraints~\eqref{LPC:each-component}, \eqref{LPC:partition-components} and \eqref{LPC:all-components} hold. So, $(\psi^*, p^*)$ is a point in $\calP$. 
	\end{proof}
\fi
	We randomly select a vertex point $(\psi, q)$ of $\calP$ such that $\E[\psi^J_{S, \beta}] = \psi^{*J}_{S, \beta} = \phi^J_{S, \beta}$ for every $J \in \calJ^\sfC, (S, \beta)$,  and $\E[q_J] = q^*_J = s_J - y_{U(J)}$ for every $J \in \calJ^\sfC$. Since $\psi$ is integral, for every $J \in \calJ$, there is a unique local solution $\big(S \subseteq U(J), \beta \in \R_{\geq 0}^{U(J)}\big)$ such that $\psi^J_{S, \beta} = 1$; we add $S$ to $S^*$ and let $\beta^*_i = \beta_i$ for every $i \in U(J)$. 
	
	This finishes the definition of the initial $S^*$ and $\beta^*$.  Let $\alpha^* = \alpha$ (recall that $\alpha_v = x_{U_v, C}$ is the demand at $v$ after the initial moving, for every $v \in R$) be the initial demand vector.  Later we shall remove facilities from $S^*$ and update $\alpha^*$ and $\beta^*$.  $S^*, \alpha^*, \beta^*$ satisfy the following properties, which will be maintained as the rounding algorithm proceeds.
	\refstepcounter{theorem} \label{properties:S-beta} 	
	\begin{properties}{properties:S-beta}
		\item $\sum_{v \in V}\alpha^*_v = \sum_{v \in V}\beta^*_v$ for every $V \in \calJ^\sfC \cup \calV^\sfN$; \label{property:supply-equal-demand}
		\item $\sum_{v \in R}\alpha^*_v = |C|$. \label{property:maintain-total-demand}
	\end{properties}
	Property~\ref{property:supply-equal-demand} is due to Properties~\ref{property:concentrated-enough-demand} and \ref{property:non-concentrated-enough-demand}. Property~\ref{property:maintain-total-demand} holds since $\sum_{v \in R}\alpha^*_v = \sum_{v \in R}x_{U_v, C} = x_{F, C} = |C|$. 	
	
%
	\subsection{The $\remove$ procedure} 
	\label{subsec:remove}
	
	In this section, we define the procedure $\remove$ that removes facilities from $S^*$ and updates $\alpha^*$ and $\beta^*$. The procedure takes a set $V \in \calJ^\sfC \cup \calV^\sfN$ as input.  If $V$ is a root black component, then we let $G = \{V\}$ be the root group containing $V$; if $V$ is a non-root concentrated component, let $G$ be the parent group of the group containing $V$; otherwise $V$ is the union of non-concentrated components in all child-groups of some group, and we let $G$ be this group.  Let $V' = \union_{J' \in G}J'$.   Before calling $\remove(V)$, we require the following properties to hold:
	\refstepcounter{theorem}\label{properties:before-remove}
	\begin{properties}{properties:before-remove}
		\item $|S^* \cap U(V)| \geq 1$; \label{property:remove-exists-one-facility}
		\item $\big|S^* \cap U(V')\big| \geq \ell - 6$.\label{property:remove-exists-ell-facilities}
	\end{properties}
	While maintaining Properties~\ref{property:supply-equal-demand} and \ref{property:maintain-total-demand}, the procedure $\remove(V)$ will
	\refstepcounter{theorem}\label{properties:after-remove}
	\begin{properties}{properties:after-remove}
		\item remove from $S^*$ exactly one open facility, which is in $U(V \cup V')$,
		\label{property:remove-one-facility}
		\item not change $\alpha^*|_{R \setminus (V \cup V')}$ and $\beta^*|_{F \setminus U(V \cup V')}$,
		 \label{property:remove-outside-nochange}
		\item increase $\alpha^*_v$ by at most a factor of $1 + O(1/\ell)$ for every $v \in V \cup V'$ and increase $\beta^*_i$ by at most a factor of $1 + O(1/\ell)$ for every $i \in U(V \cup V')$.
		\label{property:remove-inside-change-small}
	\end{properties}
	Moreover, 
	\begin{properties}[4]{properties:after-remove}
		\item the moving cost for converting the old $\alpha^*$ to the new $\alpha^*$ is at most $O(\ell^2)\beta^*_{i^*} L(J)$ for some black component $J \subseteq V$ and facility $i^* \in U(J)$; \label{property:remove-transfering-small}
		\item for every $V'' \in \calJ^\sfC \cup \calV^\sfN$, 
		$\EMD_{V''}\big(\alpha^*|_{V''}, \beta^*|_{U(V'')}\big)$ will be increased by at most a factor of $1+O(1/\ell)$.
		\label{property:remove-EMD-change-small}
	\end{properties}
	
\ifdefined\proceeding
	Due to the page limit, we only highlight the key ideas used to implement $\remove(V)$ and leave the formal description to the full version of the paper. 
\else
	Before formally describe the procedure $\remove(V)$, we first highlight some key ideas. 
\fi
	Assume $V$ is not a root component. We choose an arbitrary facility $i \in S^* \cap U(V)$. Notice that there are $\Omega(\ell)$ facilities in $S^* \cap U(V')$. If the $\beta^*_i \leq \sum_{v' \in V'}\alpha^*_{v'}/\ell$, then we can shut down $i$ and send the demands that should be sent to $i$ to $V'$. We only need to increase the supplies in $U(V')$ by a factor of $1+O(1/\ell)$.  Otherwise, we shall shut down the facility $i' \in S^* \cap U(V')$ with the smallest $\beta^*_{i'}$ value.  Since there are at least $\Omega(\ell)$ facilities in $U(V')$, we can satisfy the $\beta^*_{i'}$ units of unsatisfied demands using other facilities in $S^* \cap U(V')$.  For this $i'$, we have $\beta^*_{i'} \leq O(1) \beta^*_i$. Thus, the total amount of demands that will be moved is comparable to $\beta^*_i$. In either case, the cost of redistributing the demands is not too big. When $V$ is a root component, we shall shut down the facility $i' \in S^* \cap U(V)$ with the smallest $\beta^*_{i'}$ value. 

\ifdefined\proceeding
\else
	We now formally describe the procedure $\remove(V)$.  We first consider the case that $V$ is not a root component. So, $V$ is either a non-root component in $\calJ^\sfC$, or the union of all non-concentrated components in all child-groups of the group $G$.  In this case, $V \cap V' = \emptyset$. Let $i \in U(V) \cap S^*$ be an arbitrary facility; due to Property~\ref{property:remove-exists-one-facility}, we can find such an $i$. Let $J \subseteq V$ be the component that contains $i$.
	
	If $a:= \beta^*_i/\sum_{v' \in V'}\alpha^*_{v'} \leq \frac{1}{\ell}$, then we shall shutdown $i$. Consider the matching $f: V \times U(V) \to \R_{\geq 0}$ between $\alpha^*|_V$ and $\beta^*|_{U(V)}$ that achieves $\EMD_V\big(\alpha^*|_V, \beta^*|_{U(V)}\big)$. (Due to Property~\ref{property:supply-equal-demand}, the total supply equals the total demand.)  For every $v \in V$, we shall move $f(v, i)$ units of demand from $v$ to $V'$.  The total amount of demand moved from $V$ to $V'$ is exactly $\sum_{v \in V}f_{v, i} = \beta^*_i$.  Every $v' \in V'$ will receive $a\alpha^*_{v'}$ units of demand.    We update $\alpha^*$ to be the new demand vector: decrease $\alpha^*_v$ by $f(v, i)$ for every $v \in V$ and scale $\alpha^*_{v'}$ by a factor of $(1+a)$ for every $v' \in V'$.   By Property~\ref{property:groups-close}, the cost of moving the demands is at most $\beta^*_i O(\ell^2) L(J)$; thus, Property~\ref{property:remove-transfering-small} holds.
	
	We remove $i$ from $S^*$,  change $\beta^*_i$ to $0$, and for every $i' \in U(V')$, scale $\beta^*_{i'}$ by a factor of $(1+a)$. For this new $\alpha^*$ and $\beta^*$ vector, $\EMD_V(\alpha^*, \beta^*)$ will not increase. $\alpha^*|_{V''}$ and $\beta^*|_{U(V'')}$ are scaled by a factor of $1+a \leq 1+1/\ell$ for every $V'' \in \calJ^\sfC \cup \calV^\sfN$ such that $V'' \subseteq V'$.  Thus Properties~\ref{property:remove-inside-change-small} and \ref{property:remove-EMD-change-small} are satisfied.  Properties~\ref{property:remove-one-facility} and \ref{property:remove-outside-nochange} are trivially true. Moreover, we maintained Properties~\ref{property:supply-equal-demand} and \ref{property:maintain-total-demand}. 
	
	Now consider the case $a > 1/\ell$. In this case, we shall remove the facility $i' \in S^* \cap U(V')$ with the smallest $\beta^*_{i'}$ value from $S^*$. Notice that we have $|S^* \cap U(V')| \geq \ell-6$ before we run $\remove(V)$ due to Property~\ref{property:remove-exists-ell-facilities}. Let $a' := \beta_{i'}/\sum_{i'' \in U(V')}\beta^*_{i''}$; so, we have $a' \leq 1/(\ell-6)$. To remove the facility $i'$, we consider the function $f$ that achieves $\EMD_{V'}(\alpha^*|_{V'}, \beta^*|_{U(V')})$. We shall redistribute the demands in $V'$ so that the new demand at $v' \in V'$ will be  $(1+a')\sum_{i'' \in U(V') \setminus \{i'\}}f(v', i'')$. We remove $i'$ from $S^*$, change $\beta^*_{i'}$ to $0$ and scale up $\beta^*_{i''}$ for all other $i'' \in U(V') \setminus \{i'\}$ by $(1+a')$.  Then, the total cost for redistributing the demands in this procedure will be at most $\beta^*_{i'}O(\ell)L(J)$, due to Property~\ref{property:black-lengths-decrease}. This is at most $O(\ell)\beta^*_i L(J)$ since $a > 1/\ell$ and $a' \leq 1/(\ell-6)$. So, Properties~\ref{property:remove-one-facility} to \ref{property:remove-EMD-change-small} are satisfied and Properties~\ref{property:supply-equal-demand} and \ref{property:maintain-total-demand} are maintained.
	
	The case where $V$ is a root component can be handled in a similar way. In this case, we have $G = \{V\}$ and $V' = V$. By Property~\ref{property:remove-exists-ell-facilities}, there are at least $\ell-6$ facilities in $S^* \cap U(V)$. Then we can remove the facility $i \in U(V) \cap S^*$ with the smallest $\beta^*_i$. Using the same argument as above, we can guarantee Properties~\ref{property:remove-one-facility} to \ref{property:remove-EMD-change-small} and maintain Properties~\ref{property:supply-equal-demand} and \ref{property:maintain-total-demand}.
\fi
		
	\subsection{Obtaining the Final Solution}  
	\label{subsec:final-solution}
	To obtain our final set $S^*$ of facilities, we call the $\remove$ procedures in some order. We consider each group $G$ using the top-to-bottom order. That is, before we consider a group $G$, we have already considered its parent group.  If $G$ is a root group, then it contains a single root component $J$. If $J \in \calJ^\sfN$, repeat the the following procedure twice: if there is some facility in $S^* \cap U(J)$ then we call $\remove(J)$.  If $J \in \calJ^\sfC$ and $q_J = 1$ then we call $\remove(J)$.   Now if $G$ is a non-leaf group, then do the following. Let $V = \union_{G' \in \Lambda_G, J \in G'\cap \calJ^\sfN}J$. Repeat the following procedure twice: if there is some facility in $S^* \cap U(V)$ then we call $\remove(V)$.  For every $G' \in \Lambda_G$ and $J \in G' \cap \calJ^\sfC$ such that $q_J = 1$ we call $\remove(J)$. 
	\begin{lemma}
		After the above procedure, we have $|S^*| \leq y_F \leq k$.
	\end{lemma}

\ifdefined\proceeding
\else
	\begin{proof}
		We first show that whenever we call $\remove(V)$, Properties~\ref{property:remove-exists-one-facility} and \ref{property:remove-exists-ell-facilities} hold.  For any concentrated component $J$ with $q_J = 1$, we have called $\remove(J)$. Notice that if $q_J = 1$, then initially we have $|S^* \cap U(J)| \geq 1$ due to Constraint~\eqref{LPC:each-component}. Due to the top-down order of considering components, and Property~\ref{property:remove-one-facility},  we have never removed a facility in $S^* \cap U(J)$ before calling $\remove(J)$.  Thus, Property~\ref{property:remove-exists-one-facility} holds.  For $V\in \calV^\sfN$, we check if $|S^* \cap U(V)| \geq 1$ before we call $\remove(V)$ and thus Property~\ref{property:remove-exists-one-facility} holds. 
		
		Now consider Property~\ref{property:remove-exists-ell-facilities}. For any non-leaf group $G$, initially, we have $\big|S^* \cap \union_{J \in G}U(J)\big| \geq \floor{\sum_{J \in G}y_{U(J)}} \geq \ell$ where the first inequality is due to Property~\ref{property:non-concentrated-S} and Constraint~\eqref{LPC:partition-components} and the second is due to Property~\ref{property:groups-non-leaf-big}.  We may remove a facility from the set when we call $\remove(V)$ for $V$ satisfying one of the following conditions: (a) $V$ is a concentrated component in $G$ or in a child group of $G$, (b) $V$ is the union of the non-concentrated components in the child-groups of $G$ or (c) $V$ contains the non-concentrated components in $G$. For case (a), we removed at most $2$ facilities due to Constraint~\eqref{LPC:q-partition}. For each (b) and (c), we remove at most $2$ facilities.  Thus, we shall remove at most $6$ facilities from $\big|S^* \cap \union_{J \in G}U(J)\big| \geq \sum_{J \in G}y_{U(J)}$. Thus, Property~\ref{property:remove-exists-ell-facilities} holds.
		
		Thus, every call of $\remove$ is successful.  For a concentrated component $J$ with $q_J = 1$, we called $\remove(J)$ once.  For each $V \in \calV^\sfN$, initially we have $|S^* \cap U(V)| \in \big\{\ceil{y_{U(V)}}, \ceil{y_{U(V)}} + 1\big\}$. Before calling $\remove(V)$, we have never removed a facility from $S^* \cap U(V)$. Thus, the number of times we call $\remove(V)$ is at least the initial value of $|S^* \cap U(V)|$ minus $y_{U(V)}$.
		Overall, the number of facilities in $S^*$ after the removing procedure is at most 
		$\sum_{J \in \calJ^\sfC}\Big(\sum_{S, \beta}\psi^J_{S, \beta}-q_J\Big) + \sum_{V \in \calV^\sfN} y_{U(V)} 
			< \sum_{J \in \calJ^\sfC}y_{U(J)} + 1 + \sum_{V \in \calV^\sfN} y_{U(V)}  = y_F + 1 \leq k + 1,$
		where the first inequality is due to Constraint~\eqref{LPC:all-components}. Since $|S^*|$ is an integer, we have that $|S^*| \leq k$. 
	\end{proof}
\fi

	By Properties~\ref{property:remove-outside-nochange} and \ref{property:remove-inside-change-small}, and Constraint~\eqref{LPC:q-partition}, our final $\beta^*_i$ is at most $1+O(1/\ell)$ times the initial $\beta^*_i$ for every $i \in V$.  Finally we have $\beta^*_i \leq (1+O(1/\ell))u_i$ for every $i \in F$.  Thus, the capacity constraint is violated by a factor of $1+\epsilon$ if we set $\ell$ to be large enough. 
	
	It remains to bound the expected cost of the solution $S^*$; this is done by bounding the cost for transferring the original $\alpha^*$ to the final $\alpha^*$, as well as the cost for matching our final $\alpha^*$ and $\beta^*$.  
	
	We first focus on the transferring cost.  By Property~\ref{property:remove-EMD-change-small}, when we call $\remove(V)$, the transferring cost is at most $O(\ell^2)\beta^*_{i^*} L(J)$ for some black component $J \subseteq V$ and $i^*$. Notice that $\beta^*_{i^*}$ is scaled by at most a factor of $(1+O(1/\ell))$, we always have $\beta^*_{i^*} \leq (1+O(1/\ell))\alpha_{U(J), C}$. So, the cost is at most $O(\ell^2)x_{U(J), C}L(J)$.  If $V$ is the union of some non-concentrated components, then this quantity is at most $O(\ell^2)\ell_2 \pi_J L(J) \leq O(\ell^2 \ell_2) D_{U(J)} \leq O(\ell^2 \ell_2) D_{U(V)}$.  We call $\remove(V)$ at most twice, thus the contribution of $V$ to the transferring cost is at most $O(\ell^2 \ell_2)D_{U(V)}$. If $V$ is a concentrated component $J$, then the quantity might be large. However, the probability we call $\remove(J)$ is $\E[q_J] = q^*_J = s_J - y_{U(J)} \leq 2\ell y_{U(J)}\pi_J/x_{U(J),C}$ if $y_{U(J)} \leq 2\ell$ and it is 0 otherwise (by Property~\ref{property:concentrated-S-expectation}). So, the expected contribution of this $V$ to the transferring cost is at most $O(\ell^2)x_{U(J),C}L(J) \times 2\ell y_{U(J)}\pi_J/x_{U(J), C} \leq O(\ell^4)\pi_J L(J) \leq O(\ell^4)D_{U(J)}$ by Lemma~\ref{lemma:x-j-times-one-minus-x-j-small}. Thus, overall, the expected transferring cost is at most $O(\ell^5)D_F  = O(\ell^5)\LP$.
	
	Then we consider the matching cost. Since we maintained Property~\ref{property:supply-equal-demand}, the matching cost is bounded by $\sum_{V \in \calJ^\sfC \cup \calV^\sfN}\EMD_V(\alpha^*|_V, \beta^*|_{U(V)})$.  Due to Property~\ref{property:remove-EMD-change-small}, this quantity has only increased by a factor of $1+O(1/\ell)$ during the course of removing facilities. For the initial $\alpha^*$ and $\beta^*$, the expectation of this quantity is at most $\sum_{J \in \calJ^\sfC}O(\ell^4)D_{U(J)} + \sum_{V \in \calV^\sfN} O(\ell^2\ell_2)D_{U(V)}$ due to Properties~\ref{property:concentrated-cost} and \ref{property:non-concentrated-cost}.  This is at most $O(\ell^5)D_F = O(\ell^5)\LP$. 
	
	We have found a set $S^*$ of at most $k$ facilities and a vector $\beta^* \in \R_{\geq 0}^F$ such that $\beta^*_i = 0$ for every $i \notin S^*$ and $\beta^*_i \leq (1+O(1/\ell))u_i$. If we set $\ell = \Theta(1/\eps)$ to be large enough, then $\beta^*_i \leq (1+\eps)u_i$.  The cost for matching the $\alpha$-demand vector and the $\beta^*$ vector is at most $O(\ell^5)\LP = O(1/\eps^5)\LP$. Thus, we obtained a $O(1/\eps^5)$-approximation for \CKM with $(1+\eps)$-capacity violation.